\documentclass[journal,twocolumn]{IEEEtran}
\usepackage{amssymb}
\usepackage{newtxtext}
\usepackage{amsmath,amsfonts}
\usepackage{algpseudocode}
\usepackage{dsfont}
\usepackage{algorithm}
\usepackage{array}
\usepackage{textcomp}
\usepackage{stfloats}
\usepackage{amsthm}
\usepackage{url}
\usepackage{graphicx}
\usepackage{subcaption}
\usepackage{epstopdf}
\graphicspath{{figures/}}
\usepackage{verbatim}
\usepackage{cite}
\usepackage[colorlinks,linkcolor=blue,citecolor=blue]{hyperref}
\usepackage{balance}
\newtheorem{lemma}{Lemma}
\usepackage{booktabs}
\usepackage[printonlyused]{acronym}
\usepackage{gensymb}
\usepackage{cases}

\renewcommand{\mathbf}{\boldsymbol}

\acrodef{FR}{flexible reflector}
\acrodef{6G}{the sixth generation wireless communications}
\acrodef{LoS}{line-of-sight}
\acrodef{NLoS}{non-line-of-sight}
\acrodef{PLS}{physical layer security}
\acrodef{FSS}{Fibonacci segmentation search}
\acrodef{FS}{Fibonacci search}
\acrodef{CES}{cross entropy search}
\acrodef{CEopt}{cross entropy optimization}
\acrodef{GMM}{Gaussian mixture model}
\acrodef{RIS}{reconfigurable intelligent surface}
\acrodef{MIMO}{multiple-input multiple-output}
\acrodef{XL-MIMO}{extremely large multiple-input and multiple-output}
\acrodef{Tx}{transmitter}
\acrodef{Rx}{receiver}
\acrodef{Eve}{eavesdropper}
\acrodef{AWGN}{additive white Gaussian noise}
\acrodef{CSI}{channel state information}
\acrodef{RCS}{radar cross section}
\acrodef{KL}{Kullback-Leibler}
\acrodef{PSO}{particle swarm optimization}
\acrodef{MRT}{maximum ratio transmission}
\acrodef{RA}{reconfigurable antenna}
\acrodef{CDF}{cumulative distribution function}
\acrodef{PDF}{probability density function}
\acrodef{DoF}{degree of freedom}
\acrodefplural{DoF}{degrees of freedom}
\acrodef{ISAC}{integrated sensing and communication}
\acrodef{AN}{artificial noise}
\acrodef{RF}{radio frequency}

\allowdisplaybreaks
\begin{document}
\title{Secure Wireless Communications with Flexible Reflectors: Joint Position and Rotation Optimization for Physical Layer Security}
\author{
		Miao~Jiang,
		Xike~Wang, 
        Yiqing~Li, 
        and Guangchi~Zhang
	\IEEEcompsocitemizethanks{M. Jiang, X. Wang and G. Zhang are with the School of Information Engineering, Guangdong University of Technology, Guangzhou 510006, China (e-mail: jmiao@gdut.edu.cn, wangxike@mails.gdut.edu.cn, gczhang@gdut.edu.cn).}
	\IEEEcompsocitemizethanks{Y. Li is with the School of Computer Science and Technology, Guangdong University of Technology, Guangzhou 510006, China (e-mail: liyiq5@gdut.edu.cn).}
}
\markboth{}%
{}

\maketitle
\begin{abstract}   
By exploiting the spatial selectivity of wireless channels, physical layer security (PLS) protects confidential wireless transmissions. Flexible reflectors (FRs), which are passive metallic plates mounted on movable platforms, offer a low-cost and low-complexity solution for reshaping the propagation environment. However, their potential for secrecy-oriented signal propagation shaping has not been explored. In this paper, we study an FR-aided secure communication over non-line-of-sight links and jointly optimize the positions and rotations of FRs to maximize the secrecy rate under geometric reflection and mutual-blocking constraints. For the single-FR case, we propose a Fibonacci segmentation search (FSS) method, which partitions the feasible region into multiple unimodal intervals and obtains the optimal solution with Fibonacci search over the unimodal intervals. To reduce the computational overhead, we also propose a low-complexity cross-entropy search (CES) method, which employs Gaussian mixture model-based adaptive importance sampling to localize promising search intervals before deterministic refinement. After deriving the blocked placement interval induced by the deployed FRs and integrating it into a sequential deployment scheme, we extend the FSS and CES methods to the multi-FR case. Simulation results show that our proposed FSS method achieves significant secrecy rate performance gains over various benchmark schemes, and our proposed CES method achieves a favorable runtime-performance tradeoff. These results indicate that FRs provide a low-cost and reconfigurable mechanism for PLS.
\end{abstract}

\section{Introduction}\label{Section:Introduction}
\bstctlcite{IEEEexample:BSTcontrol}
\IEEEPARstart{D}ue to the broadcast nature of wireless propagation, secure transmission remains a fundamental challenge in modern wireless communication systems. As sensitive traffic continues to proliferate and the \ac{6G} emerge, confidentiality requirements become increasingly stringent \cite{mao2023security}. Although conventional cryptographic techniques remain indispensable, their high computational complexity and key management overhead render them inadequate for highly dynamic and latency-sensitive wireless networks~\cite{yang20265g}. \Ac{PLS} exploits the inherent randomness and spatial selectivity of wireless channels to enhance confidentiality by creating channel quality differences between legitimate receivers and eavesdroppers, without fully relying on upper-layer encryption~\cite{niu2026survey}. The fundamental issue of \ac{PLS} lies in how to reshape signal propagations in a flexible and easily implementable manner, to realize channel differentiation between the legitimate receivers and eavesdroppers. Traditional methods include secure beamforming~\cite{hu2025secure,luo2024joint}, \ac{AN} injection~\cite{niu2026survey}, and cooperative jamming~\cite{zhang2015cooperative}.

Towards the upcoming \ac{6G}, \ac{XL-MIMO} is considered as a promising candidate technology \cite{Lu2024tutorial}. By deploying massive antenna arrays, \ac{XL-MIMO} harnesses spatial diversity to boost channel capacity and transmission reliability via beamforming and spatial multiplexing \cite{wang2024tutorial}. \cite{guo2026effective} proposed an \ac{AN}-aided low-complexity beamforming scheme with closed-form power allocation to ensure \ac{PLS} in near-field \ac{XL-MIMO} systems, showing that beamfocusing yields significant \ac{PLS} gains over far-field communications. Nevertheless, \ac{XL-MIMO} suffers from critical practical limitations: the large number of antenna elements incurs prohibitive hardware costs and high power consumption \cite{zhi2024performance}, while high-dimensional signal processing yields formidable computational overhead for channel estimation and transmit beamforming, making computational complexity a major obstacle to practical implementations \cite{wang2024extremely}. These limitations make it appealing to seek alternatives that can reshape the propagation environment without relying on large-scale active arrays. 

Along this line, \ac{RIS} has emerged as a promising alternative paradigm for \ac{6G}~\cite{wu2021intelligent}. Composed of numerous passive reflecting elements with electronically tunable electromagnetic responses, the RIS actively reshapes wireless channels by collaboratively adjusting the phase, amplitude, and polarization of incident electromagnetic waves without requiring active \ac{RF} transceivers \cite{liu2021reconfigurablea,zhang2023active}. Deployed in blocked areas or cell edges, the RIS constructs directional reflected beams to mitigate severe path loss under \ac{NLoS} conditions, and suppresses multiuser interference as well as wiretapping signals via destructive scattering. It has attracted widespread attention for its application in diverse scenarios, including non-orthogonal multiple access systems~\cite{liu2022reconfigurable,li2021joint}, \ac{PLS}~\cite{cui2019securea,khoshafa2025risassisted}, and \ac{ISAC}~\cite{chu2024jointa, yu2024active}. Nevertheless, large-scale RIS-aided communication systems still suffer from complicated channel modeling, heavy channel signaling overhead, and increased hardware costs incurred by massive phase control units~\cite{wu2021intelligent}.

Different from the electrical reconfiguration paradigm discussed above, another class of schemes focused on mechanical controllability to exploit spatial \acp{DoF}. \Acp{RA}~\cite{zhao2026reconfigurable}, such as movable antennas~\cite{zhu2024movable,zhu2026tutorial}, fluid antennas~\cite{wu2026fluid}, pinching antennas~\cite{xu2026generalized}, and rotatable antennas~\cite{zheng2026rotatable}, can actively and dynamically adjust operating parameters, such as spatial position, array configuration, and coverage orientation, to meet user requirement in different scenarios. %
\Acp{RA} is capable of dynamically shortening the distance between communication nodes, bypassing obstacles, and reconfiguring array geometric layouts, which makes it suitable for \ac{ISAC} \cite{ding2026movable,zhou2024fluid,li2026joint}, \ac{PLS}~\cite{ding2025movable,dai2025rotatable,wu2026fluida,zhu2026pinchingantenna}, and other scenarios.

Following the idea of exploiting spatial \acp{DoF}, \acp{FR}, which are passive metallic reflectors mounted on controllable mechanical platforms, are recently introduced to offer a promising alternative \cite{yu2023wireless, lu2025wireless}. Passive metallic reflectors have long been employed to redirect electromagnetic waves in microwave \ac{NLoS} links and telescope reflector design. To be specific, passive metallic reflectors are used for shadowing mitigation in outdoor microwave \ac{NLoS} communications \cite{brady1952reflection}, beam control of reflector antennas \cite{chang2005manufacture}, radio signal focusing for radio telescopes \cite{shvalb2015flexiblerobotic}, and coverage enhancement for indoor blocked mmwave links \cite{balanis2016antenna}. %

Compared with \ac{RIS}, \acp{FR} can be deployed, removed, and repositioned at a low hardware cost, without active \ac{RF} circuitry or external power supply for signal manipulation. Furthermore, unlike \ac{RIS} and \ac{XL-MIMO} that require continuous reconfiguration according to instantaneous \ac{CSI}, each FR only relies on the deployment positions of communication nodes. By adjusting its position and rotation angle, the \acp{FR} can reflect signals toward the desired directions. Once configured, the reflectors remain static during operation and require no real-time \ac{CSI} acquisition or control signaling, which substantially reduce the operational overhead.

The existing studies on \acp{FR} have established a general reflection model for metal reflectors of arbitrary size and orientation \cite{yu2023wireless}, and investigated joint position and rotation design for coverage enhancement and received-power maximization \cite{lu2025wireless}. However, their potential for secrecy-oriented propagation shaping remains unexplored. Since the reflected field of an \ac{FR} is highly directional and concentrates energy around the main beam \cite{yu2023wireless}, jointly controlling its position and rotation allows the main beam to be aligned toward the legitimate receiver while steering the eavesdropper toward a low-power sidelobe. This yields a new geometric \ac{DoF} for secrecy enhancement, which is fundamentally distinct from the phase-level reconfiguration of \ac{RIS}.

Motivated by these observations, we study the \ac{FR}-aided secure communication system, where the controller adjusts the \ac{FR} positions and rotation angles according to the locations of the legitimate receiver and the eavesdropper. Our objective is to maximize the secrecy rate subject to geometric feasibility and mutual-blocking constraints. To the best of our knowledge, this is the first work to exploit the joint position and rotation reconfigurability of \acp{FR} for \ac{PLS}. The main contributions of this paper are summarized as follows:
\begin{itemize}[\IEEEsetlabelwidth{12}]
\item{We establish an \ac{FR}-aided secure communication model for \ac{NLoS} wireless networks, where the positions and rotations of the \acp{FR} can be flexibly controlled to steer the reflected beam. Based on the \ac{FR} reflection model, we formulate secrecy rate maximization problems for both single-\ac{FR} and multi-\ac{FR} systems under feasible-reflection and mutual-blocking constraints.}
\item{Due to the nonconvex and multimodal characteristics of the secrecy rate function, we develop a \ac{FSS} for the single-\ac{FR} case firstly, which partitions the feasible region into multiple unimodal intervals and obtains the solution with Fibonacci search over the unimodal intervals. }
\item{To reduce the computational cost, we further propose a \ac{CES} method for the single-FR case, which employs \ac{GMM}-based adaptive importance sampling to localize promising search intervals before deterministic refinement.}
\item{After deriving the blocked placement interval induced by each deployed \ac{FR} and incorporating it into a sequential deployment scheme, we extend both \ac{FSS} and \ac{CES} methods to the multi-\ac{FR} case.}
\item{Simulation results demonstrate that our proposed FSS achieves significant secrecy rate performance gains over various benchmark schemes, and our proposed \ac{CES} method offers a favorable runtime-performance tradeoff. }
\end{itemize}

The remainder of this paper is organized as follows. Section~\ref{section:model} presents the system model and formulates the \ac{FR}-aided secrecy rate maximization problem. Section~\ref{sec:single_FR} studies the single-\ac{FR} case and develops both the \ac{FSS} method and the \ac{CES} method. Section~\ref{sec:multi_FR} extends the design to multi-\ac{FR} systems by incorporating blocked placement intervals for both \ac{FSS} and \ac{CES} methods. Section~\ref{section:simulation} reports simulation results. Section~\ref{section:conclusion} concludes the paper.

\emph{Notations:} Scalars are denoted by italic letters. Vectors and matrices are denoted by boldface lowercase and uppercase letters, respectively. For a real number $x$, $[x]^+$ denotes $\max(x,0)$. $\boldsymbol{x}^\top$ denotes the transpose of vector $\boldsymbol{x}$. The operator $\mathbb{E}(\cdot)$ represents statistical expectation. The operator $\times$ denotes the Cartesian product of two sets. $|\cdot|$ and $\left\|\cdot\right\|$ denote the magnitude of a complex scalar and the Euclidean norm of a vector, respectively. 
$\mathcal{N}(\mathbf{x};\boldsymbol{\mu},\boldsymbol{\Sigma})
\triangleq \frac{1}{(2\pi)^{D/2}\,\lvert\boldsymbol{\Sigma}\rvert^{1/2}}
\exp\!\Bigl(-\frac{1}{2}(\mathbf{x}-\boldsymbol{\mu})^{\top}
\boldsymbol{\Sigma}^{-1}(\mathbf{x}-\boldsymbol{\mu})\Bigr)$
denotes the \ac{PDF} of a multivariate Gaussian variable $\mathbf{x}\in\mathbb{R}^{D}$ with mean vector $\boldsymbol{\mu}$ and covariance matrix $\boldsymbol{\Sigma}$.
$\mathcal{CN}(0,\sigma^2)$ denotes the circularly symmetric complex Gaussian distribution with zero mean and variance $\sigma^2$.

\section{System Model and Problem Formulation}\label{section:model}

\begin{figure}[!t]
\centering
\includegraphics[width=\linewidth]{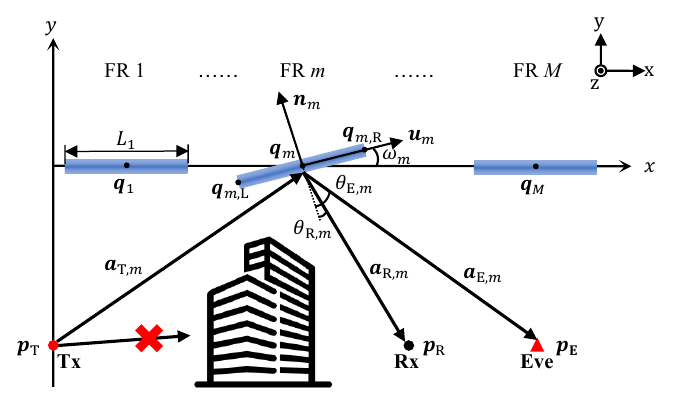}
\caption{Geometry model of the FR-aided secure communication system.}
\label{fig:model_1}
\end{figure}

As shown in Fig.~\ref{fig:model_1}, we consider a multi-\ac{FR} aided secure wireless communication system, which consists of a single-antenna \ac{Tx}, a single-antenna legitimate receiver \ac{Rx}, a single-antenna \ac{Eve}, and $M$ \acp{FR}. The positions of \ac{Tx}, \ac{Rx}, and \ac{Eve} are denoted by $ \boldsymbol{p}_\mathrm{T} = \left[ x_\mathrm{T}, y_\mathrm{T}, 0 \right]^{\top} $, $ \boldsymbol{p}_\mathrm{R} = \left[ x_\mathrm{R}, y_\mathrm{R}, 0 \right]^{\top} $, and $ \boldsymbol{p}_\mathrm{E} = \left[ x_\mathrm{E}, y_\mathrm{E}, 0 \right]^{\top} $, respectively. The direct \ac{LoS} path between \ac{Tx} and \ac{Rx} is blocked. \ac{Eve} is located on the same side as \ac{Rx} and attempts to intercept the information transmitted from \ac{Tx} to \ac{Rx}.  $M$ \acp{FR} are deployed to establish a secure communication link from \ac{Tx} to \ac{Rx}, while limiting information leakage to \ac{Eve}.

Each \ac{FR} has a size of $ L_1 \times L_2 $, where $L_1$ is its length along the $x$-axis direction and $L_2$ is its width along the $y$-axis direction. All the \acp{FR} are assumed to be rotatable and movable. For analytical tractability, we assume that the \acp{FR} are deployed on the $x$-axis within the interval $[x_{\min} , x_{\max}]$. The center of the $m$-th \ac{FR} is denoted by $ \boldsymbol{q}_m = \left[ x_m, 0, 0 \right]^{\top} $, for $ m \in  \mathcal{M} \triangleq \{ 1,\cdots ,M\}$. Let the left and right endpoints of the \ac{FR} be $\boldsymbol{q}_{m,\mathrm{L}} = [x_{m,\mathrm{L}}, y_{m,\mathrm{L}}, 0]^{\top} = [x_m - (L_1/2)\cos \omega_m, - (L_1/2)\sin \omega_m, 0]^{\top}$ and $\boldsymbol{q}_{m,\mathrm{R}} = [x_{m,\mathrm{R}}, y_{m,\mathrm{R}}, 0]^{\top} = [x_m + (L_1/2)\cos \omega_m, (L_1/2)\sin \omega_m, 0]^{\top}$, respectively, where $\omega_m$ denotes the angle between the $m$-th \ac{FR} and the $x$-axis. By sign convention, a clockwise rotation corresponds to $\omega_m < 0$, whereas a counterclockwise rotation corresponds to $\omega_m > 0$. The rotation of the \(m\)-th \ac{FR} is constrained within the interval $\omega_m \in [\omega_{\min}, \omega_{\max}]$. The corresponding tangential vector $\boldsymbol{u}_m$ and normal vector $\boldsymbol{n}_m$ are given by 
$\boldsymbol{u}_m = \left[ \cos\omega_m, \sin\omega_m, 0 \right]^{\top}$ and
$\boldsymbol{n}_m = \left[ -\sin\omega_m, \cos\omega_m, 0 \right]^{\top}$, respectively.

Let \(\boldsymbol{a}_{\mathrm{T},m}\) denote the incident vector from \ac{Tx} to the center of the $m$-th \ac{FR}, which is expressed as
\begin{equation}\label{eq:transmitted_unit_vector}
\boldsymbol{a}_{\mathrm{T},m} = \boldsymbol{q}_m - \boldsymbol{p}_\mathrm{T} = [x_{m}-x_{\mathrm{T}},-y_{\mathrm{T}},0]^{\top}.
\end{equation} Similarly, the reflected vector from the center of the $m$-th \ac{FR} to \ac{Rx} or \ac{Eve} is denoted by 
\begin{equation}\label{eq:received_reflection_unit_vector}
\boldsymbol{a}_{k,m} = \boldsymbol{p}_k - \boldsymbol{q}_m = [x_{k}-x_m, y_{k},0]^{\top}, \forall k \in \{\mathrm{R}, \mathrm{E}\}.
\end{equation} 
The distance from \ac{Tx} to the \ac{FR}, from the \ac{FR} to \ac{Rx}, and from the \ac{FR} to \ac{Eve} are $d_{\mathrm{T},m} = \| \boldsymbol{a}_{\mathrm{T},m} \|$, $d_{\mathrm{R},m} = \| \boldsymbol{a}_{\mathrm{R},m} \|$, and $d_{\mathrm{E},m} = \| \boldsymbol{a}_{\mathrm{E},m} \|$, respectively.

According to the operating principle of the \ac{FR}, effective reflections require \ac{Tx} and \ac{Rx} to lie on the same side of the \ac{FR}, i.e.,
\begin{equation}\label{eq:reflection_constraint_condition}
\boldsymbol{a}_{\mathrm{T},m}^\top\boldsymbol{n}_m\boldsymbol{n}_m^\top\boldsymbol{a}_{\mathrm{R},m} \leq  0,  \  \forall m \in \mathcal{M}.
\end{equation}

For downlink transmission, the received signal at \ac{Rx} or \ac{Eve} can be given by 
\begin{equation}\label{eq:received_signal_Rx_and_Eve}
y_k = h_k s + z_k,\quad k \in \{\mathrm{R}, \mathrm{E}\},
\end{equation}
where $s$ denotes the information symbol with $\mathbb{E}[|s|^2]=P_\mathrm{T}$, $P_\mathrm{T}$ is the transmit power of the Tx, and $z_k\sim\mathcal{CN}(0, \sigma_k^2)$ is the additive white Gaussian noise at \ac{Rx} or \ac{Eve}. According to \cite{yu2023wireless,lu2025wireless}, the channel gain at \ac{Rx} or \ac{Eve} can be represented as
\begin{equation}
|h_k|^2=\sum_{m=1}^M\frac{ \varsigma_{k,m} \lambda_{\mathrm{c}}^2 }{{(4\pi)^3}d_{\mathrm{T},m}^2 d_{k,m}^2}, \forall k \in \{\mathrm{R}, \mathrm{E}\},
\end{equation}
where $\lambda_{\mathrm{c}}$ is the carrier wavelength, $\varsigma_{k,m}$ is the \ac{RCS} of the $m$-th \ac{FR} for the signal reflected toward \ac{Rx} or \ac{Eve}, which can be expressed as
\begin{equation}\label{eq:rcs}
\begin{aligned}
\varsigma_{k,m} = \varsigma_{\mathrm{max}} \eta_{k,m}\operatorname{sinc}^2\left(\frac {\pi L_1 \Delta_{k,m}}{\lambda_{\mathrm{c}}}\right), , \forall k \in \{\mathrm{R}, \mathrm{E}\}.
\end{aligned}
\end{equation}
In \eqref{eq:rcs}, $\varsigma_{\mathrm{max}} = {4\pi L_1^2 L_2^2}/{(\lambda_{\mathrm{c}})^2}$ denotes the maximum achievable \ac{RCS}. $\eta_{k,m}$ denotes the attenuation factor of the $m$-th \ac{FR} for reflecting signal to \ac{Rx} or \ac{Eve}, which can be expressed as 
\begin{equation}\label{eq:attenuation_factor}
\eta_{k,m} = \cos^2 \theta_{k,m} = {\left(\hat{\boldsymbol{a}}_{k,m}^\top \boldsymbol{n}_m\right)}^2 , \forall k \in \{\mathrm{R}, \mathrm{E}\} , 
\end{equation}
with $\theta_{k,m}$ being the reflection angle between the reflected vector $\boldsymbol{a}_{k,m}$ and the normal vector $\boldsymbol{n}_m$ of the $m$-th \ac{FR}, as shown in Fig.\ref{fig:model_1}. $\hat{\boldsymbol{a}}_{k,m} = \frac{\boldsymbol{a}_{k,m}}{\|\boldsymbol{a}_{k,m}\|}$ denotes the unit direction vector of $\boldsymbol{a}_{k,m}$.  $\Delta_{k,m}$ denotes the tangential-direction mismatch between the incident and reflected propagation directions. This parameter measures the mismatch between the actual reflected direction and the ideal specular-reflection direction, which can be represented as
\begin{equation}
    \Delta_{k,m}=\hat{\boldsymbol{a}}_{k,m}^\top\boldsymbol{u}_m - \hat{\boldsymbol{a}} _{\mathrm{T},m}^\top\boldsymbol{u}_m.
\end{equation}
Therefore, the achievable rate $R_k$ of \ac{Rx} or \ac{Eve}  can be given as
\begin{equation}\label{eq:rate_Rx_and_Eve}
R_k = \log_2\left( 1 + \frac{P_\mathrm{T}|h_k|^2}{\sigma_k^2} \right), \quad k \in \{\mathrm{R}, \mathrm{E}\}.
\end{equation} 
Accordingly, the secrecy rate $R_{\mathrm{s}}$ can be defined as
\begin{equation}\label{eq:rate_secure}
R_{\mathrm{s}} = \left[ R_{\mathrm{R}} - R_{\mathrm{E}} \right]^+.
\end{equation}
In this paper, we propose to jointly optimize the position and rotation of \acp{FR} to maximize the secrecy rate, subject to the deployment, rotation and reflection constraints. Therefore, the secrecy rate maximization problem can be formulated as
\begin{subequations}\label{eq:P1}
\begin{align}
\max_{\{x_m, \omega_m\}_{m=1}^M} & \, R_{\mathrm{s}}  \label{subeq:P1_question} \\
 \mathrm{s.t.} \quad & |x_m - x_n| \geq d_{\min}, \quad \forall \, m, n \in \mathcal{M}, \, m \neq n, \label{subeq:P1_distance}\\
 & \boldsymbol{a}_{\mathrm{T},m}^\top \boldsymbol{n}_m\boldsymbol{n}_m^\top\boldsymbol{a}_{\mathrm{R},m} \leq  0, \forall  m \in \mathcal{M}, \label{subeq:P1_reflection}\\
 & x_{\min} \leq x_m \leq x_{\max}, \forall  m \in \mathcal{M}, \label{subeq:P1_x_range}\\
 & \omega_{\min} \leq \omega_m \leq \omega_{\max}, \forall  m \in \mathcal{M}.\label{subeq:P1_omega_range}
\end{align}
\end{subequations}

In \eqref{subeq:P1_question}, the objective is to maximize the secrecy rate. Constraint \eqref{subeq:P1_distance} enforces a minimum spacing \( d_{\min} \) between \acp{FR} to avoid mutual interference, while \eqref{subeq:P1_reflection} guarantees geometric feasibility of the reflection. Constraints \eqref{subeq:P1_x_range} and \eqref{subeq:P1_omega_range} specify the feasible ranges of the position and angle of each \ac{FR}, respectively.

Problem \eqref{eq:P1} is highly nonconvex due to the $\operatorname{sinc}^2(\cdot)$ term, the fractional structure of the objective, and the coupled geometry constraints, which makes it intractable for off-the-shelf solvers. To address this, we firstly propose an \ac{FSS} method and a low-complexity \ac{CES} method for the single-\ac{FR} case in Section~\ref{sec:single_FR}. Subsequently, we extend our proposed \ac{FSS} and \ac{CES} methods to multiple \acp{FR} in Section~\ref{sec:multi_FR}.

\section{Joint Position and Rotation Optimization for Single-FR Case} \label{sec:single_FR}
\subsection{Proposed FSS Method}
For the single-\ac{FR} case, problem \eqref{eq:P1} reduces to
\begin{subequations}\label{eq:P2}
\begin{align}
\max_{x_1, \omega_1} \quad  & \tilde{R}_{\mathrm{s}}(x_1, \omega_1)  \label{subeq:P2_question} \\
\mathrm{s.t.} \quad &  \boldsymbol{a}_{\mathrm{T},1}^\top \boldsymbol{n}_1\boldsymbol{n}_1^\top\boldsymbol{a}_{\mathrm{R},1} \leq  0,\label{subeq:P2_reflection} \\
 & x_{\min} \leq x_1 \leq x_{\max}, \label{subeq:P2_x_range}\\
 & \omega_{\min} \leq \omega_1 \leq \omega_{\max}, \label{subeq:P2_omega_range}
\end{align}
\end{subequations}
where $\tilde{R}_{\mathrm{s}}(x_1, \omega_1) = {(\beta_{\mathrm{R}}\sigma^2_{\mathrm{E}}+\sigma^2_{\mathrm{R}}\sigma^2_{\mathrm{E}})}/{(\beta_{\mathrm{E}} \sigma_{\mathrm{R}}^2+\sigma^2_{\mathrm{E}}\sigma^2_{\mathrm{R}})}$, $\beta_k = P_\mathrm{T} |h_k|^2$, for $k \in \{\mathrm{R}, \mathrm{E}\}$.

\begin{figure}[!htbp]
\centering
\begin{subfigure}{0.95\linewidth}
\centering
\includegraphics[width=\linewidth]{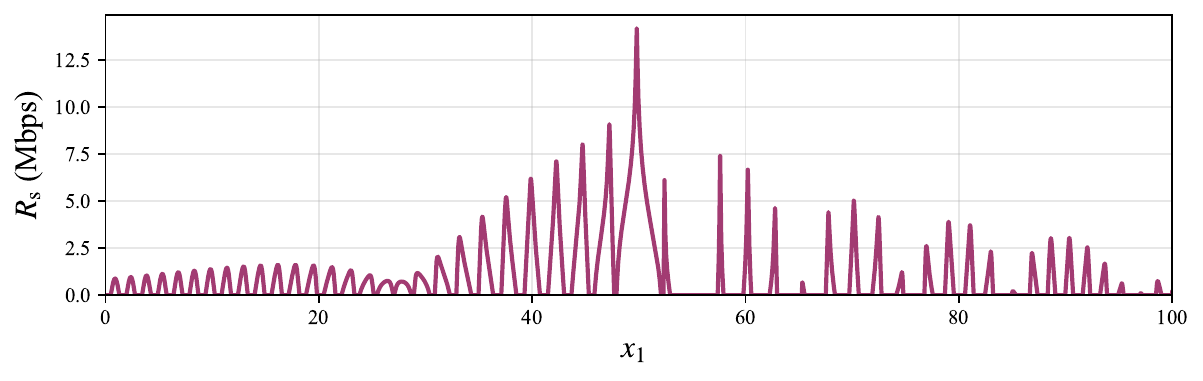}
\caption{Secrecy rate, $R_{\mathrm{s}}$, versus the position, $x_1$, where $\omega_1=0^\circ$.}
\label{fig:sub:a}
\end{subfigure}

\begin{subfigure}{0.95\linewidth}
\centering
\includegraphics[width=\linewidth]{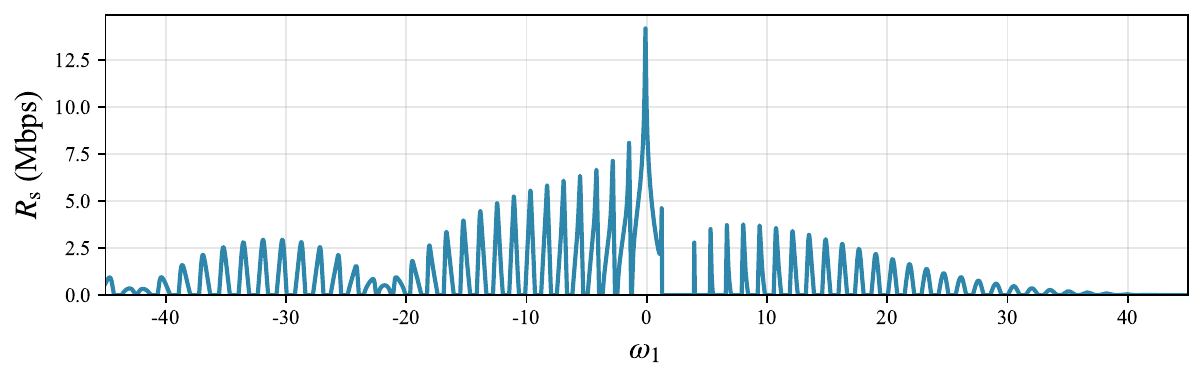}
\caption{Secrecy rate, $R_{\mathrm{s}}$, versus the rotation, $\omega_1$, where $x_1=50$ m.}
\label{fig:sub:b}
\end{subfigure}

\caption{Illustration of the nonconvex and multimodal characteristics of secrecy rate $R_{\mathrm{s}}$.}
\label{fig:Non-convex-diagram}
\end{figure}
Fig.~\ref{fig:Non-convex-diagram} illustrates the nonconvex and multimodal characteristics of the secrecy rate $R_{\mathrm{s}}$ under the system parameters: \ac{Tx}, \ac{Rx}, and \ac{Eve} are located at $[0,-50,0]^\top$, $[100,-50,0]^\top$, and $[110,-50,0]^\top$, respectively, the transmit power is $P_\mathrm{T} = 30$ dBm, the carrier frequency is $2.4$ GHz, the bandwidth is $1$ MHz, the size of the \ac{FR} is $L_1 \times L_2 = 1.25$ m $\times$ $0.65$ m, the feasible region is $x_1 \in [0, 100]$ m and $\omega_1 \in [-45\degree, 45\degree]$. In Fig.~\ref{fig:Non-convex-diagram}(a), $R_{\mathrm{s}}$ is plotted versus the position $x_1$ for a fixed angle $\omega_1 = 0\degree$, whereas in Fig.~\ref{fig:Non-convex-diagram}(b), $R_{\mathrm{s}}$ is plotted versus the rotation angle $\omega_1$ for a fixed position $x_1 = 50$ m. The nonconvex structure of $\operatorname{sinc}^2(\cdot)$ function, together with the angle- and position-dependent attenuation factor $\eta_{k,1}$, produce multiple local maxima along both axes, so that each one-dimensional profile is multimodal rather than unimodal. As a result, conventional one-dimensional search methods, such as golden-section or Fibonacci search, are prone to being trapped in a local optimum, which motivate the segmented search method developed below.

To address the nonconvex and multimodal issues presented in Fig.~\ref{fig:Non-convex-diagram}, we propose to transform problem \eqref{eq:P2} into a two-layer hierarchical search problem, which is formulated as%
\begin{subequations}\label{eq:P4}%
\begin{align}%
\max_{x_1} \quad &\hat{R}_{\mathrm{s}}(x_1)\label{subeq:P4_question} \\
\mathrm{s.t.} \quad  & x_{\min} \leq x_1 \leq x_{\max}. \label{subeq:P4_x_range}
\end{align}
\end{subequations}
Here, $\hat{R}_{\mathrm{s}}(x_1)$ is the maximum achievable secrecy rate obtained by optimizing the rotation angle $\omega_1$ over its feasible range with a given position $x_1$, which is the optimum of the following problem%
\begin{subequations}\label{eq:P3}%
\begin{align}%
\max_{\omega_1}\quad &\tilde{R}_{\mathrm{s}}\left(x_1,\omega_1\right) \label{subeq:P3_question} \\
\mathrm{s.t.}\quad &\boldsymbol{a}_{\mathrm{T},1}^\top \boldsymbol{n}_1\boldsymbol{n}_1^\top\boldsymbol{a}_{\mathrm{R},1} \leq  0,\label{subeq:P3_reflection} \\
&\omega_{\min} \leq \omega_1 \leq \omega_{\max}.\label{subeq:P3_omega_range}
\end{align}
\end{subequations}

For constraint \eqref{subeq:P3_reflection}, when $x_1$ is fixed, it can be transformed into a constraint with respect to the range of $\omega_1$ using the following Lemma~\ref{lemma2}.

\begin{lemma}\label{lemma2}
When $x_1$ is fixed, \eqref{subeq:P3_reflection} can be equivalently expressed as
\begin{equation}\label{eq:lemma2}
\begin{aligned}
-\arccos\left( \frac{\boldsymbol{a}_{\mathrm{R},1}^\top \boldsymbol{e}_x}{d_{\mathrm{R},1}} \right) \leq \omega_1 \leq \arccos \left( \frac{\boldsymbol{a}_{\mathrm{T},1}^\top \boldsymbol{e}_x}{d_{\mathrm{T},1}} \right),
\end{aligned}
\end{equation}
where $\boldsymbol{e}_x$ denotes the unit vector along the $x$-axis.
\end{lemma}
\begin{IEEEproof}
Please refer to Appendix \ref{proof of the lemma 2}.
\end{IEEEproof}

By applying Lemma~\ref{lemma2}, \eqref{subeq:P3_reflection} can be simplified to \eqref{eq:lemma2}, which is only influenced by $\omega_1$. By integrating \eqref{eq:lemma2} with \eqref{subeq:P3_omega_range}, problem \eqref{eq:P3} can be equivalently expressed as
\begin{subequations}\label{eq:P6}
\begin{align}
\max_{\omega_1}\quad &\tilde{R}_{\mathrm{s}}\left(x_1, \omega_1\right) \label{subeq:P6_question} \\
\mathrm{s.t.}\quad & \bar{\omega}_{\min} \leq \omega_1 \leq \bar{\omega}_{\max},\label{subeq:P6_omega_range}
\end{align}
\end{subequations}
where $\bar{\omega}_{\min} = \max(\omega_{\min}, -\arccos({\boldsymbol{a}_{\mathrm{R},1}^\top \boldsymbol{e}_x}/{d_{\mathrm{R},1}}))$ and $\bar{\omega}_{\max} = \min(\omega_{\max}, \arccos({\boldsymbol{a}_{\mathrm{T},1}^\top \boldsymbol{e}_x}/{d_{\mathrm{T},1}}))$.

\begin{algorithm}[!htbp]
\caption{Proposed \ac{FSS} Method for Optimizing a Single-Variable Multimodal Function}
\begin{algorithmic}[1]
\State \textbf{Input}: $f(t)$, $t_{\min}$, $t_{\max}$,  $N$, $L$, $\varepsilon$.
\Statex /* Segmenting into unimodal subintervals */
\State $\bar{t}_1 = t_{\min}$, $\delta_t = \dfrac{t_{\max}-t_{\min}}{L}$, $u=1$.\label{alg1:step-divide-interval-begin}
\For{$i = 2$ \textbf{to} $L$}
    \State $\bar{t}_i = \bar{t}_{i-1} + \delta_t$, $\bar{t}_{i+1} = \bar{t}_i + \delta_t$.
    \If{$f(\bar{t}_i)-f(\bar{t}_{i-1})\ge 0$ \textbf{and} $f(\bar{t}_i)-f(\bar{t}_{i+1})\ge 0$}
        \State $\mathcal{S}_u = [\bar{t}_{i-1},\, \bar{t}_{i+1}]$.
        \State $u = u+1$.
    \EndIf
\EndFor\label{alg1:step-divide-interval-end}
\Statex /* Searching within each unimodal subinterval */
\For{$j = 1$ \textbf{to} $u-1$}
 \Statex \quad /* Fibonacci search */
    \State $a_0 = \min \mathcal{S}_j,\quad b_0 = \max \mathcal{S}_j$. \label{alg1:step-fibonacci-search-begin}
    \For{$i = 1$ \textbf{to} $N-1$}\label{alg1:step-fibonacci-update-begin}
        \State Compute $\rho_i = 1-\frac{F_{N-i+1}}{F_{N-i+2}}$ and $l_i = b_{i-1} - a_{i-1}$.
        \State Update $c_i = a_{i-1} + \rho_i  l_i$ and $d_i = b_{i-1} - \rho_i  l_i$.
        \If{$f(c_i) \ge f(d_i)$}
            \State $a_{i} = a_{i-1},\quad b_{i} = d_i$.
        \Else
            \State $b_{i} = b_{i-1},\quad a_{i} = c_i$.
        \EndIf
        \If{$b_{i} - a_{i} < \varepsilon$}
            \State \textbf{break}
        \EndIf
    \EndFor\label{alg1:step-fibonacci-update-end}
    \State Compute $t_j^* = \frac{a_i + b_i}{2}$. \label{alg1:step-fibonacci-search-end}
\EndFor
\State \textbf{Output}: $t^\dagger= \underset{t\in\{t_1^*,\ldots,t_{u-1}^*\}}{\arg\max}\, f(t)$.
\end{algorithmic}\label{alg:alg1}
\end{algorithm}
To solve \eqref{eq:P4} and \eqref{eq:P6}, we propose the Fibonacci segmented search (FSS) method as detailed in Algorithm~\ref{alg:alg1}. 

In Algorithm~\ref{alg:alg1}, $f(t)$ denotes the objective function to be optimized, $t_{\min}$ and $t_{\max}$ correspond to the lower and upper bounds of the optimization variable $t$, respectively. 

In Steps \ref{alg1:step-divide-interval-begin}--\ref{alg1:step-divide-interval-end} of Algorithm \ref{alg:alg1}, the multimodal objective function is partitioned into multiple unimodal intervals firstly. All valid unimodal intervals are then filtered and preserved, with each such interval recorded as $\mathcal{S}_u$.

In Steps \ref{alg1:step-fibonacci-search-begin}--\ref{alg1:step-fibonacci-search-end} of Algorithm \ref{alg:alg1}, the Fibonacci search based one-dimensional search is executed to obtain the local optimal solution $t_j^*$ within each unimodal interval $\mathcal{S}_j \triangleq [a_0,b_0]$. $N$ denotes the maximum number of iterations and is selected such that the $N$-th Fibonacci number satisfies $F_N \geq \frac{b_0 - a_0}{\varepsilon}$, where $F_n$ denotes the $n$-th term of the Fibonacci sequence, defined as 
\begin{equation}\label{eq:fibonacci_sequence}
\begin{aligned}
F_n = F_{n-1} + F_{n-2}, \; n \geq 2, 
\end{aligned}
\end{equation}
with $F_0 = F_1 = 1$.

Steps \ref{alg1:step-fibonacci-update-begin}--\ref{alg1:step-fibonacci-update-end} of Algorithm \ref{alg:alg1} define the Fibonacci search procedure. In the $i$-th iteration, $\rho_i$ and $l_i$ denote the step size and the interval length, respectively. Firstly, candidate update endpoints $c_i$ and $d_i$ are computed. Then, $f(c_i)$ and $f(d_i)$ are compared to determine which endpoint should be updated for the subsequent iteration. 

Finally, the optimal solution $t^{\dagger}$ is selected from all local optima $t_j^*$ obtained from each unimodal interval $\mathcal{S}_j$.

For simplicity, we define the Algorithm~\ref{alg:alg1} as
\begin{equation}\label{eq:FS_function}
\begin{aligned}
t^\dagger = \mathrm{FSS}\left(f(t), t_{\min}, t_{\max}\right).
\end{aligned}
\end{equation}

Therefore, problem \eqref{eq:P4} can be solved as
\begin{equation}\label{eq:x_optimal}
\begin{aligned}
x_1^\dagger = \mathrm{FSS}\left(\hat{R}_{\mathrm{s}}(x_1), x_{\min}, x _{\max}\right),
\end{aligned}
\end{equation}
where $\hat{R}_{\mathrm{s}}(x_1)$ can be expressed as
\begin{equation}\label{eq:hat_R_s}
\begin{aligned}
\hat{R}_{\mathrm{s}}(x_1)=\max _{\omega_1 \in [\bar{\omega}_{\min}, \bar{\omega}_{\max}]} \tilde{R}_{\mathrm{s}}\left(x_1, \omega_1\right).
\end{aligned}
\end{equation}
To solve problem \eqref{eq:hat_R_s}, the FSS method is adopted to optimize $\tilde{R}_\mathrm{s}(\omega_1;x_1)$ within $[\bar{\omega}_{\min}, \bar{\omega}_{\max}]$ with fixed $x_1$, which is formulated as
\begin{equation}\label{eq:omega_optimal}
\begin{aligned}
\omega_1^\dagger = \mathrm{FSS}\left(\tilde{R}_{\mathrm{s}}(\omega_1;x_1), \bar{\omega}_{\min}, \bar{\omega}_{\max}\right).
\end{aligned}
\end{equation}
To summarize, our proposed FSS method for solving single-FR deployment problem \eqref{eq:P2} is detailed in Algorithm~\ref{alg:alg2}. 

\begin{algorithm}[!htbp]
\caption{Proposed FSS Method for Solving Single-FR Deployment Problem \eqref{eq:P2}}
\begin{algorithmic}[1]
\State \textbf{Input}:$\tilde{R}_s(x,\omega)$, $x_{\min}, x_{\max}$, $\omega_{\min}, \omega_{\max}$, $L$, $\varepsilon$.
\State $\bar{x}_1 = x_{\min}$, $\delta_x = \dfrac{x_{\max}-x_{\min}}{L}$, $u = 1$.
\For{$i = 2$ \textbf{to} $L$}
    \State $\bar{x}_i = \bar{x}_{i-1} + \delta_x$, $\bar{x}_{i+1} = \bar{x}_i + \delta_x$.
    \If{$\hat{R}_s(\bar{x}_i) \geq \hat{R}_s(\bar{x}_{i-1})$ and $\hat{R}_s(\bar{x}_i) \geq \hat{R}_s(\bar{x}_{i+1})$}
        \State $\bar{\mathcal{S}}_u = \left[\bar{x}_{i-1}, \bar{x}_{i+1}\right]$.
        \State $u = u + 1$.
    \EndIf
\EndFor
\For{$j = 1$ \textbf{to} $u-1$}
\State Obtain $ x_{1,j}^*$ from $\bar{\mathcal{S}}_j$ via steps~\ref{alg1:step-fibonacci-search-begin}--\ref{alg1:step-fibonacci-search-end} of Algorithm~\ref{alg:alg1}.
\EndFor
\State $x^\dagger_1 = \underset{x\in\{x_{1,j}^*\}_{j = 1}^{u-1}}{\arg\max}\, \hat{R}_\mathrm{s}(x)$.
\State $\omega^\dagger_1 = \underset{\omega_1 \in [\bar{\omega}_{\min}, \bar{\omega}_{\max}]}{\arg\max} \tilde{R}_{\mathrm{s}}\left(x_1^\dagger, \omega_1\right)$.
\State \textbf{Output}:
\parbox[t]{\dimexpr\linewidth-\algorithmicindent\relax}{ $[x^\dagger_1, \omega^\dagger_1]^\top$.} %
\end{algorithmic}\label{alg:alg2}
\end{algorithm}

\subsection{Proposed Low-Complexity CES Method}

Although the proposed \ac{FSS} method can solve problem \eqref{eq:P2} by traversing the entire feasible region to obtain a favorable solution, it suffers from high computational complexity and long search time. Therefore, in this subsection, a low-complexity \ac{CES} method is proposed to solve problem \eqref{eq:P2}.

For simplicity, we define $\boldsymbol{z} = [x_1, \omega_1]^\top$ and the corresponding search domain $\mathcal{A} = \mathcal{B}_1 \times \mathcal{W}$,
with $\mathcal{B}_1 = [x_{\min},x_{\max}]$ and $\mathcal{W} = [\omega_{\min},\omega_{\max}]$ denoting the feasible ranges of $x_1$ and $\omega_1$, respectively. To recast constraint \eqref{subeq:P2_reflection}, we employ the following penalty function 
\begin{equation}\label{eq:CES_objective_function_h}
\begin{aligned} 
h (\boldsymbol{z}) = R_{\mathrm{s}} - \lambda [g(\boldsymbol{z})]^+,
\end{aligned}
\end{equation}
where $g(\boldsymbol{z}) = \boldsymbol{a}_{\mathrm{T},1}^\top \boldsymbol{n}_1\boldsymbol{n}_1^\top\boldsymbol{a}_{\mathrm{R},1}$. When $g(\boldsymbol{z}) < 0$, the constraint holds and no penalty is incurred; otherwise, a penalty term is added. The penalty coefficient $\lambda$ is chosen large enough to ensure the constraint is satisfied at convergence.

The cross-entropy optimization method is an adaptive importance sampling technique originally developed for estimating rare-event probabilities and subsequently generalized to deterministic optimization \cite{rubenstein2004crossentropy}. Its core principle is to iteratively refine a parameterized probability distribution so that the sampling mass concentrates around the optimal region of the search space. Let $p_{\boldsymbol{\tau}}(\boldsymbol{z})$ denote the parameterized sampling distribution, over the search domain $\mathcal{A}$, and let $q(\boldsymbol{z})$ denote the target distribution whose probability mass concentrates on the near-optimal region of the search space. The cross-entropy method seeks the parameter $\boldsymbol{\tau}^*$ that minimizes the \ac{KL} divergence between $q(\boldsymbol{z})$ and $p_{\boldsymbol{\tau}}(\boldsymbol{z})$, yielding the following reformulation of problem \eqref{eq:P2}
\begin{equation}\label{eq:CEopt_problem}
\boldsymbol{\tau}^* = \arg\min_{\boldsymbol{\tau}} D_{\mathrm{KL}}\bigl(q(\boldsymbol{z}) \,\|\, p_{\boldsymbol{\tau}}(\boldsymbol{z})\bigr).
\end{equation}
Since the target distribution $q(\boldsymbol{z})$ is unknown, the cross-entropy method relies on elite samples to approximate it. The proposed \ac{CES} method iteratively adapts the sampling distribution so that its probability mass concentrates on regions yielding near-optimal objective values. The detailed steps of our proposed \ac{CES} method for solving single-FR deployment problem \eqref{eq:CEopt_problem} are presented in the following.

Firstly, the proposed \ac{CES} method generates $L_{\mathrm{CE}}$ samples from $\mathcal{A}$. Since the objective function is multimodal, the \ac{GMM} is adopted as the sampling distribution, which is composed of $Q$ Gaussian distribution components with means initialized uniformly over the solution space. For each component, $L_{\mathrm{CE}}/Q$ feasible points are sampled. At the $t$-th iteration, the \ac{PDF} of the \ac{GMM} is given by
\begin{equation}\label{eq:GMM_PDF}
    f_{\mathrm{GMM}} (\mathbf{x};\boldsymbol{\tau}^t ) = \sum_{q=1}^{Q} \eta_q^t \mathcal{N} (\mathbf{x};\boldsymbol{\mu}_q^t, \boldsymbol{\Sigma}_q^t),
\end{equation}
where $\boldsymbol{\tau}^t = \left\{ \eta_q^t;\ \boldsymbol{\mu}_q^t;\ \boldsymbol{\Sigma}_q^t \right\}_{q=1}^Q$ denotes the distribution parameters of the \ac{GMM}. $\eta_q^t$, $\boldsymbol{\mu}_q^t$, and $\boldsymbol{\Sigma}_q^t$ are the weight, mean vector, and covariance matrix of the $q$-th Gaussian component at the $t$-th iteration, respectively.

After sampling, the objective function $h(\hat{\mathbf{z}}^t_i)$ of the $i$-th sample $\hat{\mathbf{z}}^t_i$ is evaluated, for $i\in \{1, 2, \dots, L_\mathrm{CE}\}$. The elite set $\mathcal{X}^{t}$ is then constructed by selecting the top $K$ samples ranked by their objective values $h(\hat{\mathbf{z}}^t_i)$, i.e., $|\mathcal{X}^{t}| = K$. Each $\hat{\boldsymbol{z}}^t_i \in \mathcal{X}^{t}$ constitutes an elite solution for subsequent parameter updates. 

Subsequently, for the $i$-th elite sample $\hat{\boldsymbol{z}}^t_i \in \mathcal{X}^{t}$, the posterior probability that sample $\hat{\boldsymbol{z}}^t_i$ belongs to the $q$-th Gaussian component is computed as
\begin{equation}\label{eq:GMM_update_gamma}
\gamma_{q,i}^t = \frac{\eta_q^{t-1} \,
\mathcal{N}(\hat{\boldsymbol{z}}^t_i;\, \boldsymbol{\mu}_q^{t-1},\, \boldsymbol{\Sigma}_q^{t-1})}
{\sum_{l=1}^{Q} \eta_l^{t-1} \,
\mathcal{N}(\hat{\boldsymbol{z}}^t_i;\, \boldsymbol{\mu}_l^{t-1},\, \boldsymbol{\Sigma}_l^{t-1})},
\; \forall q \in \mathcal{Q},
\end{equation}
where $\mathcal{Q} \triangleq \{1, 2, \dots, Q\}$. Given the posterior probabilities, the \ac{PDF} parameters of the \ac{GMM} at the $t$-th iteration are updated via maximum likelihood as
\begin{subnumcases}{\label{eq:GMM_update}}
\eta_q^t = \frac{1}{K} \sum_{i=1}^{K} \gamma_{q,i}^t, \\
\boldsymbol{\mu}_q^t = \frac{\sum_{i=1}^{K} \gamma_{q,i}^t \, \hat{\boldsymbol{z}}^t_i}
{\sum_{i=1}^{K} \gamma_{q,i}^t}, \\
\boldsymbol{\Sigma}_q^t = \frac{\sum_{i=1}^{K} \gamma_{q,i}^t \,
(\hat{\boldsymbol{z}}^t_i - \boldsymbol{\mu}_q^t)(\hat{\boldsymbol{z}}^t_i - \boldsymbol{\mu}_q^t)^\top}
{\sum_{i=1}^{K} \gamma_{q,i}^t}.
\end{subnumcases}

After $T_{\mathrm{iter}}$ iterations, the estimated means $\{\boldsymbol{\mu}_q^{T_{\mathrm{iter}}}\}_{q=1}^{Q}$ concentrate near the local maxima of $h(\boldsymbol{z})$. Therefore, the \ac{FSS} method is applied to each estimated Gaussian component, using a refined search interval of half-width $\Delta_x$ centered at the mean, i.e.,
\begin{equation}\label{eq:CES_search_interval}
\mathcal{G}_q = [x_{q,\min},\, x_{q,\max}], \quad \forall q \in \mathcal{Q},
\end{equation}
where $x_{q,\min} = \max(\mu_{x,q}^{T_{\mathrm{iter}}} - \Delta_x, x_{\min})$ and $x_{q,\max} = \min(\mu_{x,q}^{T_{\mathrm{iter}}} + \Delta_x, x_{\max})$. After applying the \ac{FSS} method to each interval $\mathcal{G}_q$, a candidate solution $[x_q^*, \omega_q^*]$ is obtained, and the final solution is selected as
\begin{equation}\label{eq:CES_get_best_x}
[x^\dagger_1, \omega^\dagger_1]^\top = \underset{[x,\omega] \in \{[x_q^*,\omega_q^*]\}_{q=1}^{Q}}{\arg\max}\; \tilde{R}_\mathrm{s}(x,\omega).
\end{equation}
To summarize, our proposed low-complexity \ac{CES} method for solving single-FR deployment problem \eqref{eq:P2} is detailed in Algorithm~\ref{alg:alg3}.

\begin{algorithm}[!htbp]
\caption{Proposed CES Method for Solving Single-FR Deployment Problem \eqref{eq:P2}}
\begin{algorithmic}[1]
\State \textbf{Input}: $Q$, $L_{\mathrm{CE}}$, $K$, $T_{\mathrm{iter}}$, $\mathcal{A}$, $\sigma_{x,0}^2$, $\sigma_{\omega,0}^2$, $\Delta_x$.
\Statex /* Initialization */
\For{$q = 1$ \textbf{to} $Q$}
\State $\mu_{x,q}^0 =x_{\min} + \frac{q-\frac{1}{2}}{Q}(x_{\max}-x_{\min})$.
\State $\mu_{\omega,q}^0 =\omega_{\min} +\frac{q-\frac{1}{2}}{Q}(\omega_{\max}-\omega_{\min})$.
\State $\eta_q^0 = \dfrac{1}{Q}$.
\EndFor
\Statex /* Cross-entropy optimization iterations */
\For{$t = 0$ \textbf{to} $T_{\mathrm{iter}}-1$}
        \State \parbox[t]{\dimexpr\linewidth-\algorithmicindent\relax}{Generate $L_{\mathrm{CE}}$ samples $\{\hat{\boldsymbol{z}}_i^t\}$ according to the \ac{GMM} \ac{PDF} \eqref{eq:GMM_PDF}.}
        \State Evaluate $h(\hat{\boldsymbol{z}}_i^t)$ for all $L_{\mathrm{CE}}$ samples.
        \State Update elite set $\mathcal{X}^t$ based on the $K$ largest $h(\hat{\boldsymbol{z}}_i^t)$.
        \State Update $\boldsymbol{\tau}^{t+1}$ via \eqref{eq:GMM_update_gamma} and \eqref{eq:GMM_update}.%
\EndFor
\State Construct subregions $\{\mathcal{G}_q\}_{q=1}^Q$ using \eqref{eq:CES_search_interval} with $\boldsymbol{\tau}^{T_{\mathrm{iter}}}$.
\Statex /* \ac{FSS} refinements */
\For{$i = 1$ \textbf{to} $Q$}
    \State \parbox[t]{\dimexpr\linewidth-\algorithmicindent\relax}{ Obtain $[x_i^*,\omega_i^*]$ by Algorithm~\ref{alg:alg2} on $\mathcal{G}_i$.}
\EndFor
\State Select the best solution $[x^\dagger_1,\omega^\dagger_1]^\top$ by \eqref{eq:CES_get_best_x}.
\State \textbf{Output}: $[x^\dagger_1,\omega^\dagger_1]^\top $.
\end{algorithmic}\label{alg:alg3}
\end{algorithm}

\section{Joint Positions and Rotations for Multi-FR Case} \label{sec:multi_FR}
In this section, we extend our proposed \ac{FSS} and \ac{CES} methods to the multi-\ac{FR} case. 
\subsection{Proposed FSS Method for Multi-FR Case}
 Compared with the single-\ac{FR} case problem \eqref{eq:P2}, \eqref{eq:P1} includes additional spacing constraints to prevent mutual blockages and interferences among \acp{FR}. Building on the single-\ac{FR} solution, we optimize \acp{FR} sequentially: after placing an \ac{FR}, we compute the interval that becomes unavailable to the remaining \acp{FR} and remove it from the feasible position set before optimizing the next \ac{FR}. The sequential optimization of multiple \acp{FR} can be formulated as a series of single-\ac{FR} optimization problems, each with a reduced feasible region. 
 
 Therefore, the key issue reduces to characterizing the blocked placement interval induced by a deployed \ac{FR} in \eqref{subeq:P1_distance}. Instead of relying merely on a fixed minimum separation constant $d_{\min}$, a geometry-dependent blocked placement interval $\mathcal{E}_m = [x_{m,\mathrm{L}}, x_{m,\mathrm{R}}]$ is derived for the $m$-th \ac{FR}. The blocked placement interval $\mathcal{E}_m$ is used to replace constraint \eqref{subeq:P1_distance}. This interval specifies the region in which no additional \ac{FR} can be placed once the $m$-th \ac{FR} has been deployed. The following Lemma~\ref{lemma3} characterizes this blocked interval and provides the geometric basis for sequential multi-\ac{FR} deployment.

\begin{lemma}\label{lemma3}
For the $m$-th deployed \ac{FR} at position $[x_m, y_m, 0 ]^\top$ with rotation angle $\omega_m$, the blocked interval on the $x$-axis is $\mathcal{E}_m = [x_{m,\mathrm{L}}, x_{m,\mathrm{R}}]$, where
\begin{align*}
x_{m,\mathrm{L}} &= \min\!\left\{ \frac{\pm\frac{L_1}{2}\| \boldsymbol{n}_{\mathrm{L}}\| - c_{0,\mathrm{L}}}{n_{\mathrm{L},x}}, x_m-d_{\min} \right\}, \\
x_{m,\mathrm{R}} &= \max\!\left\{ \frac{\pm\frac{L_1}{2}\| \boldsymbol{n}_{\mathrm{R}}\| - c_{0,\mathrm{R}}}{n_{\mathrm{R},x}},  x_m+d_{\min} \right\}, \\
\boldsymbol{n}_{\mathrm{L}} &= \begin{cases}
[y_{m,\mathrm{L}} - y_{\mathrm{R}}, x_{\mathrm{R}} - x_{m,\mathrm{L}}, 0]^{\top}, & \omega_m > 0\\
[y_{m,\mathrm{L}} - y_{\mathrm{T}}, x_{\mathrm{T}} - x_{m,\mathrm{L}}, 0]^{\top}, & \omega_m < 0
\end{cases}, \\
\boldsymbol{n}_{\mathrm{R}} &= \begin{cases}
[y_{m,\mathrm{R}} - y_{\mathrm{R}}, x_{\mathrm{R}} - x_{m,\mathrm{R}}, 0]^{\top}, & \omega_m > 0\\
[y_{m,\mathrm{R}} - y_{\mathrm{T}}, x_{\mathrm{T}} - x_{m,\mathrm{R}}, 0]^{\top}, & \omega_m < 0
\end{cases}, \\
c_{0,\mathrm{L}} &= \begin{cases}
x_{m,\mathrm{L}} y_{\mathrm{R}} - x_{\mathrm{R}} y_{m,\mathrm{L}}, & \omega_m > 0\\
x_{m,\mathrm{L}} y_{\mathrm{T}} - x_{\mathrm{T}} y_{m,\mathrm{L}}, & \omega_m < 0
\end{cases}, \\
c_{0,\mathrm{R}} &= \begin{cases}
x_{m,\mathrm{R}} y_{\mathrm{R}}-x_{\mathrm{R}} y_{m,\mathrm{R}}, & \omega_m > 0\\
x_{m,\mathrm{R}} y_{\mathrm{T}}-x_{\mathrm{T}} y_{m,\mathrm{R}}, & \omega_m < 0
\end{cases},
\end{align*}
$n_{\mathrm{L},x}$ and $n_{\mathrm{R},x}$ are the first components of $\boldsymbol{n}_{\mathrm{L}}$ and $\boldsymbol{n}_{\mathrm{R}}$, respectively.
\end{lemma}
\begin{IEEEproof}
Please refer to Appendix \ref{proof of the lemma 3}.
\end{IEEEproof}

Within the available placement range, we use Algorithm~\ref{alg:alg2} to optimize the \acp{FR} sequentially. After placing the first \ac{FR}, we compute its blocked interval $\mathcal{E}_1=[x_{1,\mathrm{L}}, x_{1,\mathrm{R}}]$ by Lemma~\ref{lemma3}, remove it from the original search space $\mathcal{B}_1$, and obtain a new feasible set $\mathcal{B}_2 = \mathcal{B}_1 \setminus \mathcal{E}_1$. The feasible region $\mathcal{B}_{m}$ for the placement of the $m$-th FR is given by 
\begin{equation}\label{eq:multi_FR_process}
\begin{aligned}
\mathcal{B}_{m} = \mathcal{B}_{m-1} \setminus \mathcal{E}_{m-1}, \quad \forall m \in \{2,3,\cdots,M\}.
\end{aligned}
\end{equation}
When $m>1$, $\mathcal{B}_{m}$ may consist of multiple disjoint subintervals. Thus, we need to search over each subinterval $\{\mathcal{B}_{m,j}\}_{j=1}^{J_m}$ within $\mathcal{B}_{m}$ via our proposed FSS method detailed in Algorithm~\ref{alg:alg1}, where $J_{m}$ denotes the number of disjoint subintervals contained in the $m$-th feasible region $\mathcal{B}_{m}$. After searching over the $j$-th subinterval $\mathcal{B}_{m,j}$ via Algorithm~\ref{alg:alg1}, locally optimal solutions $[x^*_{m,j}, \omega^*_{m,j}]^\top$ can be obtained. The optimal solution is selected from all candidate local optimal solutions as
\begin{equation}\label{eq:FSS_multi_get_best_x}
[x^\dagger_m, \omega^\dagger_m]^\top = \underset{[x,\omega]^\top \in \{ [x^*_{m,j}, \omega^*_{m,j}]^\top\}_{j=1}^{J_m} } {\arg\max} \tilde{R}_\mathrm{s}(x,\omega).
\end{equation}

To summarize, our proposed \ac{FSS} method for solving multi-FR deployment problem \eqref{eq:P1} is presented in Algorithm~\ref{alg:alg4}.

\begin{algorithm}[!htbp]
\caption{Proposed FSS Method for Solving Multi-FR Deployment Problem \eqref{eq:P1}}
\label{alg:alg4}
\begin{algorithmic}[1]
\State \textbf{Input}: $\mathcal{B}_1$, $\omega_{\min}$, $\omega_{\max}$, $L$, $M$, $\varepsilon$.
\For{$m = 1$ \textbf{to} $M$}
    \State \parbox[t]{\dimexpr\linewidth-\algorithmicindent\relax}{Split the disjoint interval $\mathcal{B}_m$ into $J_m$ continuous subintervals as $\{\mathcal{B}_{m,j}\}_{j=1}^{J_m}$.
    \label{step:use-alg2-for-multi-FR}}
    \For{$j = 1$ \textbf{to} $J_m$}
        \State \parbox[t]{\dimexpr\linewidth-\algorithmicindent-2em\relax}{       
        Obtain $[x^*_{m,j}, \omega^*_{m,j}]$ via Algorithm~\ref{alg:alg2} over $\mathcal{B}_{m,j}$.}
    \EndFor
    \State Obtain $[x^\dagger_m, \omega^\dagger_m]$ via \eqref{eq:FSS_multi_get_best_x}.
        \State \parbox[t]{\dimexpr\linewidth-\algorithmicindent\relax}{Calculate $\mathcal{E}_m = [x_{m,\mathrm{L}}, x_{m,\mathrm{R}}]$ via Lemma~\ref{lemma3} with $[x^\dagger_m, \omega^\dagger_m]^\top$.}\vspace{0.1em}
    \State Update $\mathcal{B}_{m+1}$ via \eqref{eq:multi_FR_process}.
\EndFor
\State \textbf{Output}: $\{[x_m^\dagger, \omega_m^\dagger]^\top\}_{m=1}^{M}$.
\end{algorithmic}
\end{algorithm}

\subsection{Proposed Low-Complexity \ac{CES} Method for Multi-FR Case}
For the multi-FR case, the \ac{CES} method can be also extended to obtain a suboptimal solution efficiently. 

Similar to the \ac{FSS} method, the \ac{CES} method for the multi-\ac{FR} case also deploys the \acp{FR} sequentially. After placing the $m$-th \ac{FR}, we compute the blocked interval $\mathcal{E}_m = [x_{m,\mathrm{L}}, x_{m,\mathrm{R}}]$ via Lemma~\ref{lemma3} and remove it from the feasible region for the next $(m+1)$-th \ac{FR}. The feasible region for the $(m+1)$-th \ac{FR} is defined as
\begin{equation}\label{eq:multi_A_CES}
\begin{aligned}
\mathcal{A}_{m+1} = \mathcal{B}_{m+1} \times \mathcal{W} , \quad \forall m \in \mathcal{M},
\end{aligned}
\end{equation}
where $\mathcal{B}_{m+1} = \mathcal{B}_{m} \setminus \mathcal{E}_m $. The penalized objective function of \ac{CES} for deploying the $m$-th FR is formulated as
\begin{equation}\label{eq:CES_objective_function_multi_h}
\begin{aligned}
\bar{h}_m({\mathbf{z}_m}) = R_{\mathrm{s}} -  \lambda_m [\bar{g}_m(\boldsymbol{z}_m)]^+,\quad \forall m \in \mathcal{M}, 
\end{aligned}
\end{equation}
where $\mathbf{z}_m = [x_m, \omega_m]^\top$, $\lambda_m$ is the penalty factor for the $m$-th \ac{FR}, and $\bar{g}_m(\boldsymbol{z}_m) = \boldsymbol{a}_{\mathrm{T},m}^\top \boldsymbol{n}_m\boldsymbol{n}_m^\top\boldsymbol{a}_{\mathrm{R},m}$.

To summarize, our proposed low-complexity \ac{CES} method for solving multi-FR deployment problem \eqref{eq:P1} is detailed in Algorithm~\ref{alg:alg5}.

\begin{algorithm}[!htbp]
\caption{Proposed CES Method for Solving Multi-FR Deployment Problem \eqref{eq:P1}}
\label{alg:alg5}
\begin{algorithmic}[1]
\State \textbf{Input}: $Q$, $L_{\mathrm{CE}}$, $K$, $T_{\mathrm{iter}}$, $\mathcal{A}_1$, $\sigma_{x,0}^2$, $\sigma_{\omega,0}^2$, $\Delta_x$, $M$.
\For{$m = 1$ \textbf{to} $M$}
    \State \parbox[t]{\dimexpr\linewidth-\algorithmicindent\relax}{Obtain $[x^\dagger_m, \omega^\dagger_m]^\top$ via Algorithm~\ref{alg:alg3} with sampling region $\mathcal{A}_m$.}
    \State \parbox[t]{\dimexpr\linewidth-\algorithmicindent\relax}{Calculate $\mathcal{E}_m = [x_{m,\mathrm{L}}, x_{m,\mathrm{R}}]$ via Lemma~\ref{lemma3} with $[x^\dagger_m, \omega^\dagger_m]^\top$.}\vspace{0.1em}
    \State Update $\mathcal{A}_{m+1}$ via \eqref{eq:multi_A_CES}.
\EndFor
\State \textbf{Output}: $\{[x_m^\dagger, \omega_m^\dagger]^\top\}_{m=1}^{M}$.
\end{algorithmic}
\end{algorithm}

\section{Simulation Results}\label{section:simulation}
In the simulations, unless otherwise specified, the transmit power is $P_\mathrm{T} = 30$ dBm, the noise power spectral density is $-174$ dBm/Hz, and the bandwidth is $B = 1$ MHz. \ac{Tx}, \ac{Rx}, and \ac{Eve} are located at $[0,-50,0]^\top$ m, $[100,-50,0]^\top$ m, and $[110,-50,0]^\top$ m, respectively. The carrier frequency is $2.4$ GHz. The \ac{FR} sizes are $L_1 = 1.25$ m and $L_2 = 0.65$ m. The minimal distance between FRs is $d_{\min} = 0.625$ m. The feasible search intervals of the \ac{FR} positions is $[x_{\min}, x_{\max}] = [0, 100]$ m and the rotation angles is $[\omega_{\min}, \omega_{\max}] = [-45^\circ, 45^\circ]$. For the proposed \ac{FSS} method, the convergence threshold is $\varepsilon = 0.001$. For the proposed \ac{CES} method, the number of elite samples is $K = 100$. The half-width of the local refinement interval is $\Delta_x = 2$ m. For multi-\ac{FR} scenarios, the number of \acp{FR} is $M = 10$. All the results are obtained by averaging over $1000$ independent Monte Carlo simulations.

\begin{figure}[!htbp]
\centering
\includegraphics[width=\linewidth]{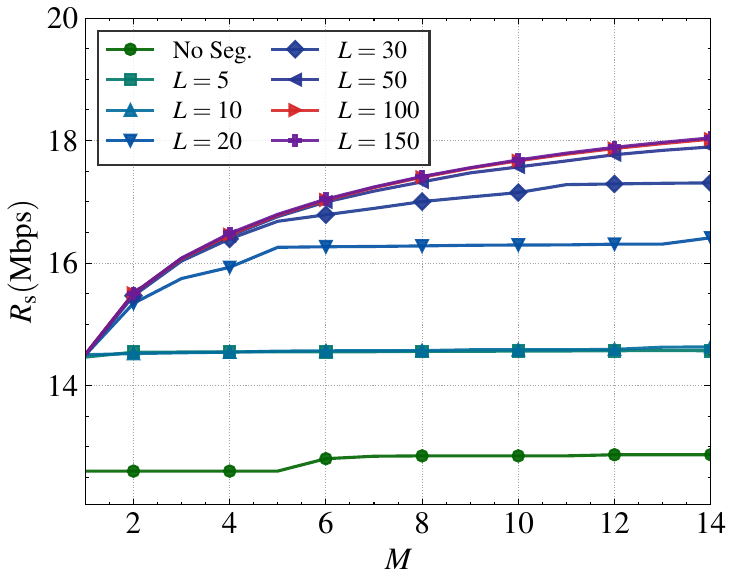}
\caption{Secrecy rate, $R_{\mathrm{s}}$, versus the number of FRs, $M$, where the number of interval partitions in our proposed \ac{FSS} method is $L=5,10,20,30,50,100,150$.}
\label{fig:FS-multi-different-L}
\end{figure}
Fig.~\ref{fig:FS-multi-different-L} presents the secrecy rate, $R_{\mathrm{s}}$, versus the number of \acp{FR}, $M$, under different numbers of interval partitions, $L$, in our proposed \ac{FSS} method. In the legend, ``No Seg.'' denotes the direct use of the Fibonacci search method over the entire domain without subintervals segmentation. As observed in Fig.~\ref{fig:FS-multi-different-L}, a small $L$ fails to sufficiently search for the global optimum among multiple local optima. As $L$ increases, the performance gap gradually narrows. In particular, when $L \ge 100$, the performance improvement becomes negligible which indicates further increasing $L$ only introduces extra computational overhead without bringing substantial performance improvement. Therefore, we set $L = 100$ in the following simulations.

\begin{figure}[!htbp]
\centering
\includegraphics[width=\linewidth]{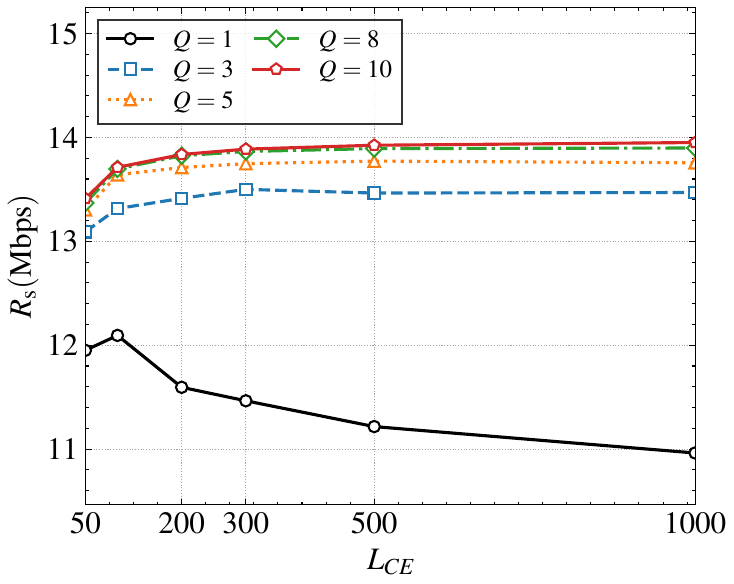}
\caption{Secrecy rate, $R_{\mathrm{s}}$, versus the number of samples, $L_{\mathrm{CE}}$, in our proposed \ac{CES} method, where the number of \ac{GMM} components is $Q=1,3,5,8,10$.}
\label{fig:CES-GMM-different-Q-LCE}
\end{figure}
In Figs.~\ref{fig:CES-GMM-different-Q-LCE} and~\ref{fig:CES-GMM-different-Q-P}, the impact of hyperparameter selection on the performance of our proposed \ac{CES} method is systematically evaluated. Fig.~\ref{fig:CES-GMM-different-Q-LCE} illustrates the secrecy rate, $R_{\mathrm{s}}$, versus the number of samples, $L_{\mathrm{CE}}$, in our proposed \ac{CES} method under different numbers of \ac{GMM} components, $Q$. 
Notably, for the degenerate case of $Q=1$, the \ac{GMM} reduces to a single Gaussian distribution, which is incapable of capturing the multimodal structure of the secrecy rate function. In this case, increasing $L_{\mathrm{CE}}$ with a fixed elite size $K$ merely concentrates the sampling distribution more aggressively around an already selected local optimum, thereby accelerating premature convergence and even degrading the achievable secrecy rate $R_{\mathrm{s}}$. This observation highlights the necessity of the multi-component \ac{GMM} adopted in our proposed \ac{CES} method.
It can be also observed that both the stability and the achieved secrecy rate $R_{\mathrm{s}}$ improve significantly as $Q$ ($Q \geq 3$) and $L_{\mathrm{CE}}$ increase. However, no further performance gain can be observed when $Q>8$, while the total computational complexity increases. Therefore, we set $Q=8$ for the following simulations. As for $L_{\mathrm{CE}}$, the secrecy rate increases as $L_{\mathrm{CE}}$ becomes larger, and gradually saturates when $L_{\mathrm{CE}} \geq 500$. Therefore, we set $L_{\mathrm{CE}}=500$ in the following simulations.

\begin{figure}[!htbp]
\centering
\includegraphics[width=\linewidth]{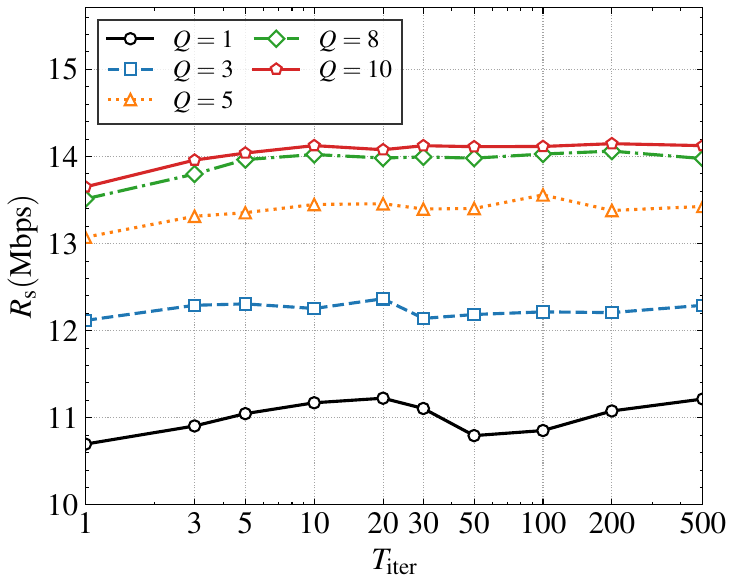}
\caption{Secrecy rate, $R_{\mathrm{s}}$, versus the number of iterations, $T_{\mathrm{iter}}$, in our proposed \ac{CES} \ method, where the number of \ac{GMM} components is $Q=1,3,5,8,10$.}
\label{fig:CES-GMM-different-Q-P}
\end{figure}
Fig.~\ref{fig:CES-GMM-different-Q-P} depicts the secrecy rate, $R_{\mathrm{s}}$, versus the number of iterations, $T_{\mathrm{iter}}$, in our proposed \ac{CES} method under different numbers of \ac{GMM} components, $Q$. Consistent with Fig.~\ref{fig:CES-GMM-different-Q-LCE}, increasing $Q$ yields a higher secrecy rate $R_{\mathrm{s}}$. Moreover, enlarging $T_{\mathrm{iter}}$ also contributes to performance improvement. $R_{\mathrm{s}}$ grows noticeably with $T_{\mathrm{iter}}$ when $T_{\mathrm{iter}}<100$, whereas the secrecy rate barely changes for $T_{\mathrm{iter}}>100$. Therefore, we set $T_{\mathrm{iter}}=100$ for the following simulations.

\begin{figure}[!htbp]
\centering
\begin{subfigure}{\linewidth}
\centering
\includegraphics[width=\linewidth]{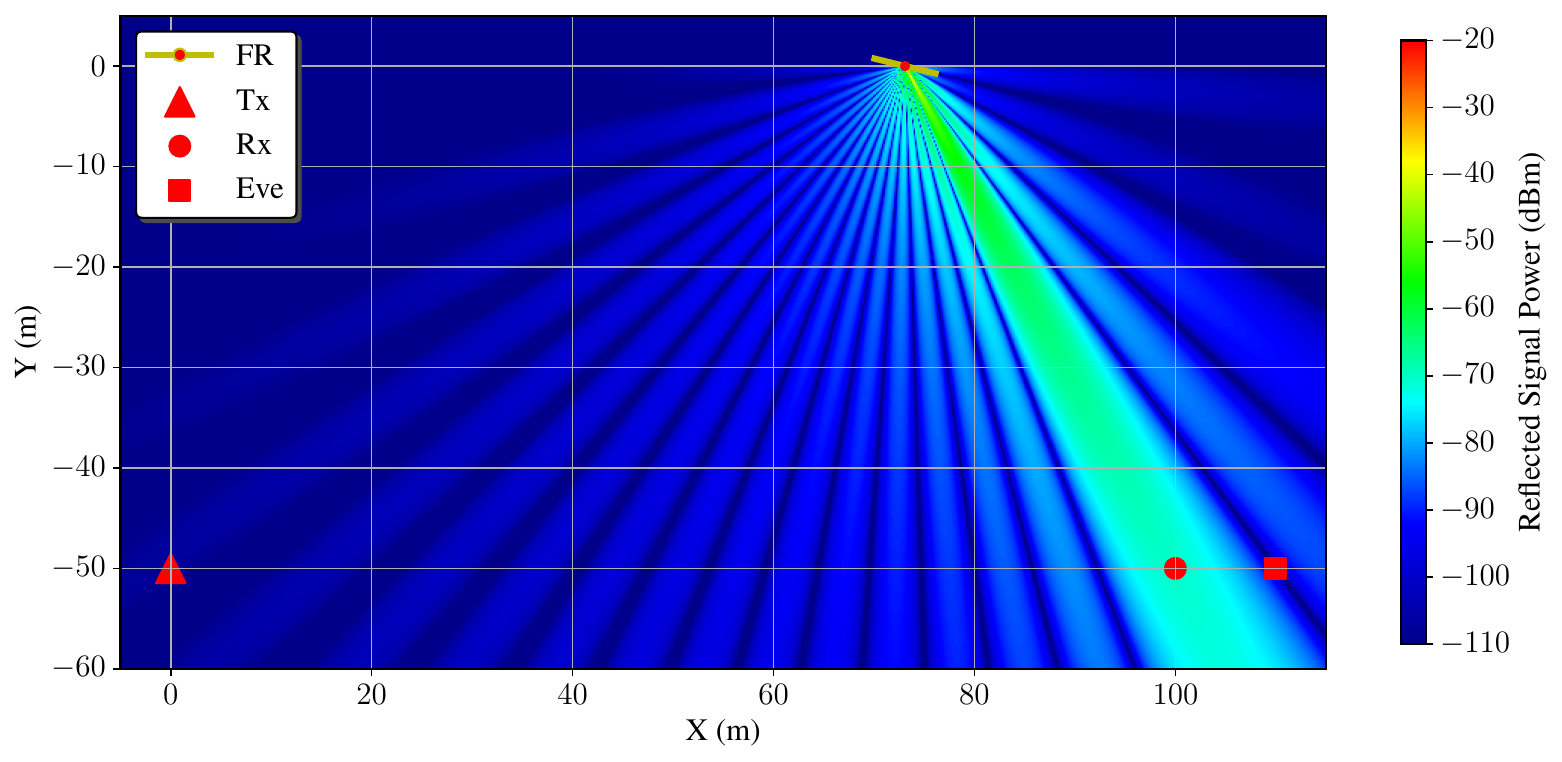}
\caption{Single-\ac{FR}}
\label{fig:sub2:a}
\end{subfigure}

\vspace{4mm} %

\begin{subfigure}{\linewidth}
\centering
\includegraphics[width=\linewidth]{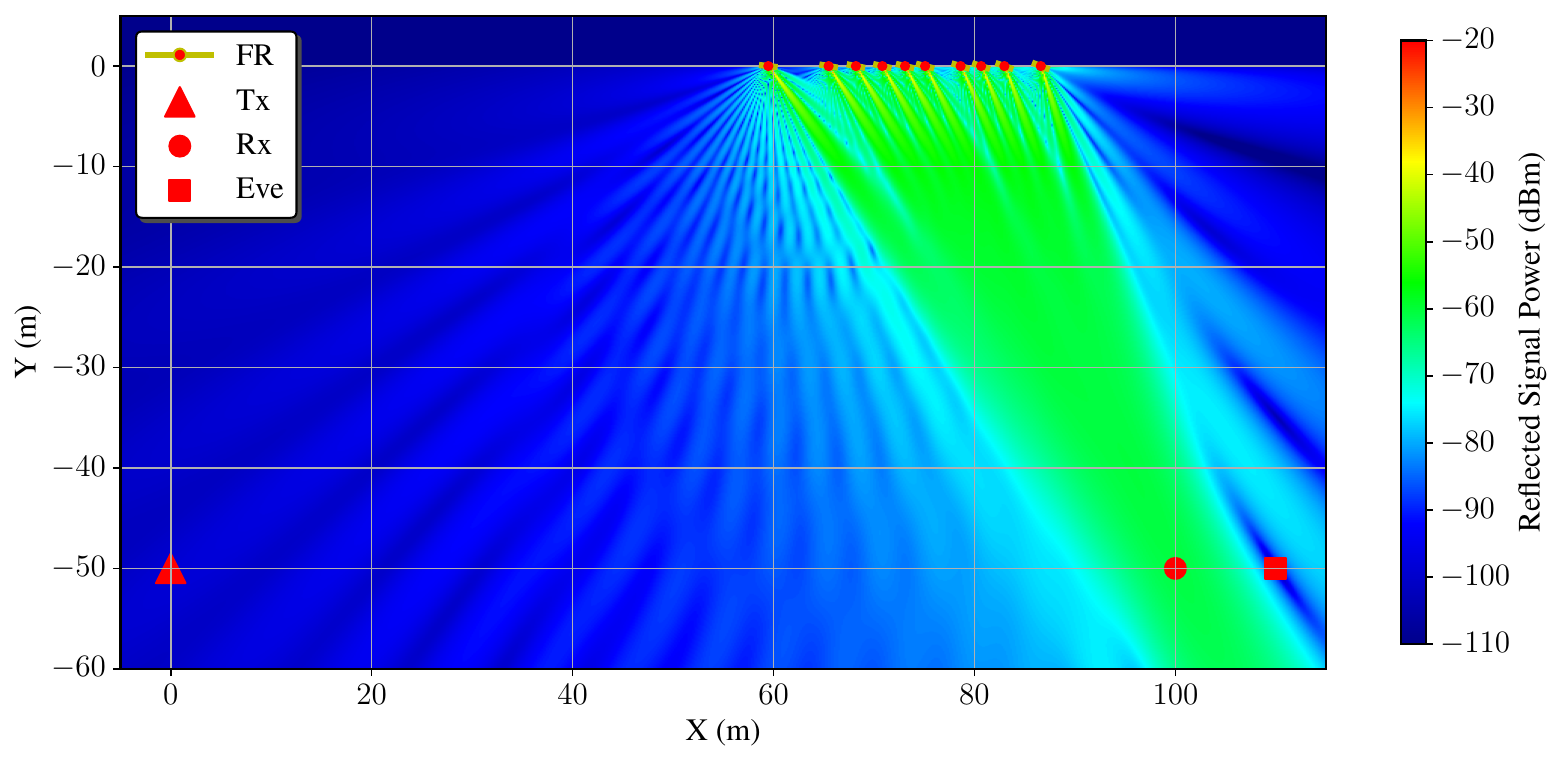}
\caption{Multi-\ac{FR} ($M= 10$)}
\label{fig:sub2:b}
\end{subfigure}
\caption{Reflected power distribution of the FR-aided secure wireless communication system.}
\label{fig:signal-distribution}
\end{figure}

Fig.~\ref{fig:sub2:a} depicts the reflected power distribution after single-\ac{FR} reflection using our proposed \ac{FSS} method. The optimized \ac{FR} steers the high-power main beam toward the legitimate receiver \ac{Rx}, while the eavesdropper \ac{Eve} is situated within the low-power region between adjacent beams. This observation means that our proposed scheme can enhance the received signal power at the legitimate link and suppress the signal strength intercepted by \ac{Eve} simultaneously. Fig.~\ref{fig:sub2:b} presents the reflected power distribution after multi-\ac{FR} deployment with $M = 10$. Consistent with the single-\ac{FR} case, the deployed FRs make \ac{Eve} locate in a broader low-power region, which can enhance the robustness of multi-\ac{FR} deployment against the position uncertainty of \ac{Eve}.

\begin{figure*}[t]
\centering
\begin{subfigure}[t]{0.3\linewidth}
    \centering
    \includegraphics[width=\linewidth]{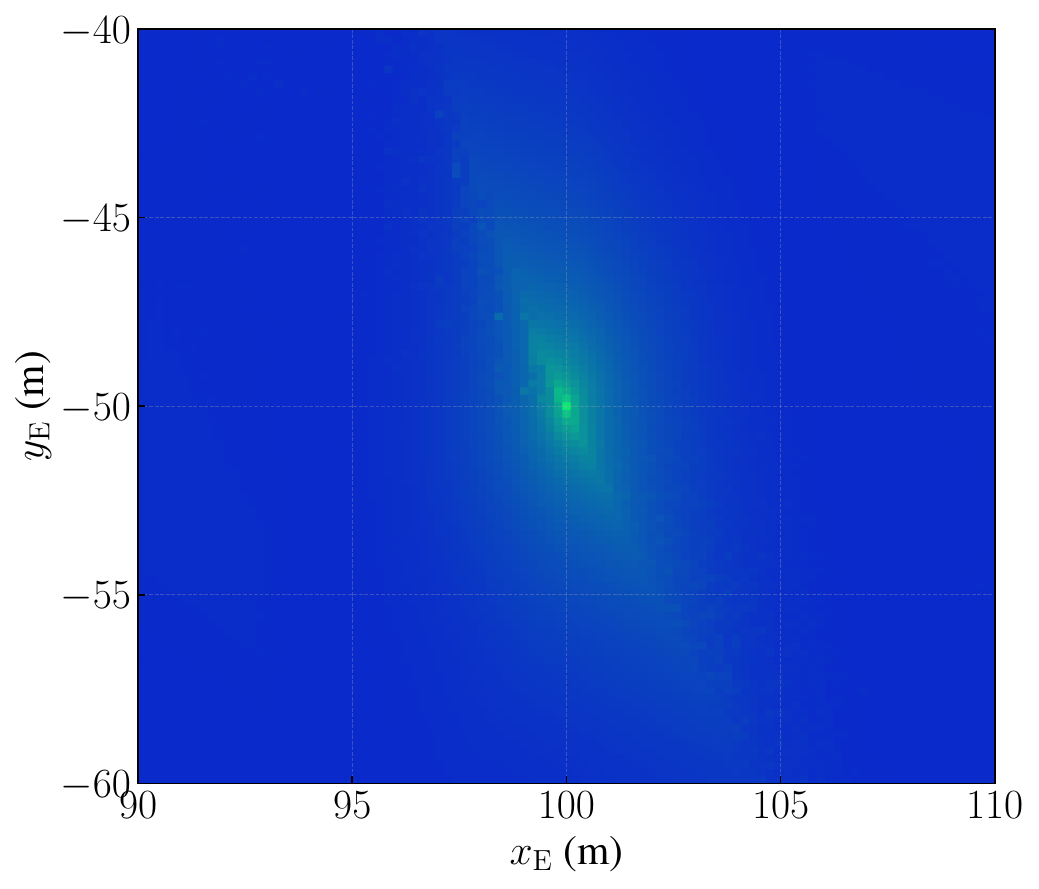}
    \caption{Optimized position and rotation}
    \label{subfig:Different-optimization-variables-a}
\end{subfigure}
\hfill
\begin{subfigure}[t]{0.3\linewidth}
    \centering
    \includegraphics[width=\linewidth]{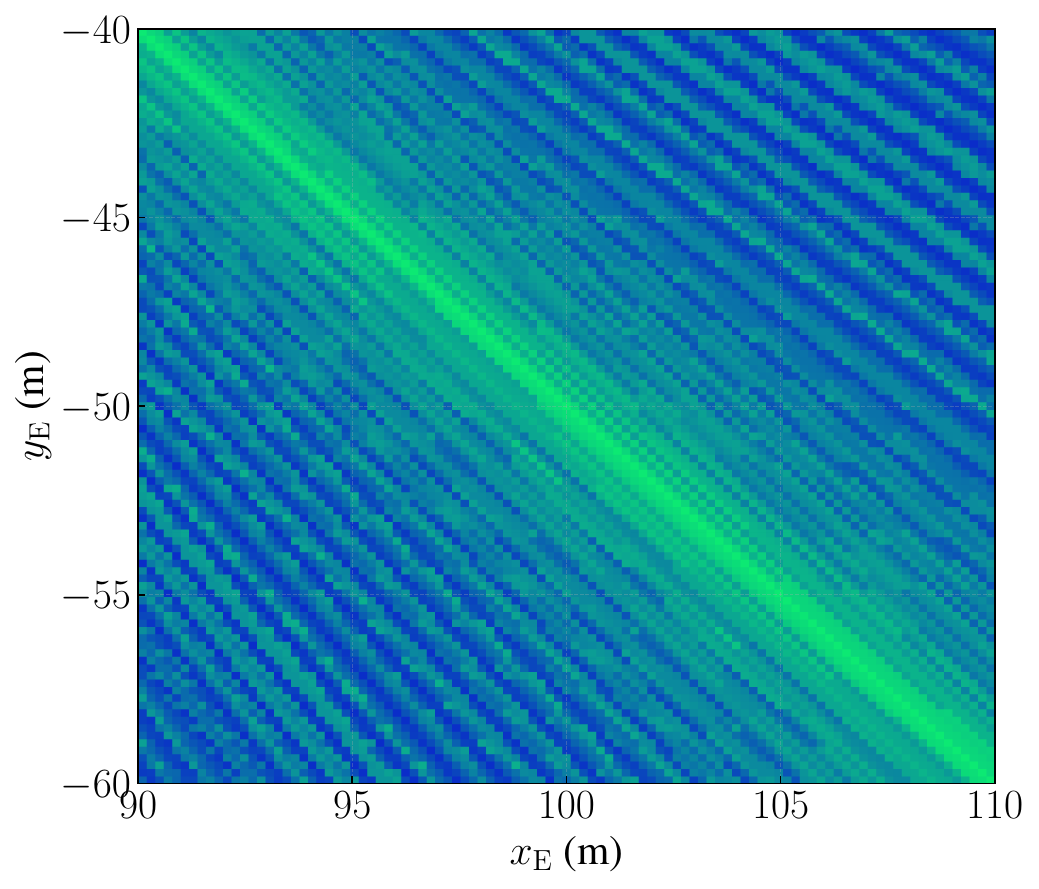}
    \caption{Fixed position and optimized rotation}
    \label{subfig:Different-optimization-variables-b}
\end{subfigure}
\hfill
\begin{subfigure}[t]{0.3\linewidth}
    \centering
    \includegraphics[width=\linewidth]{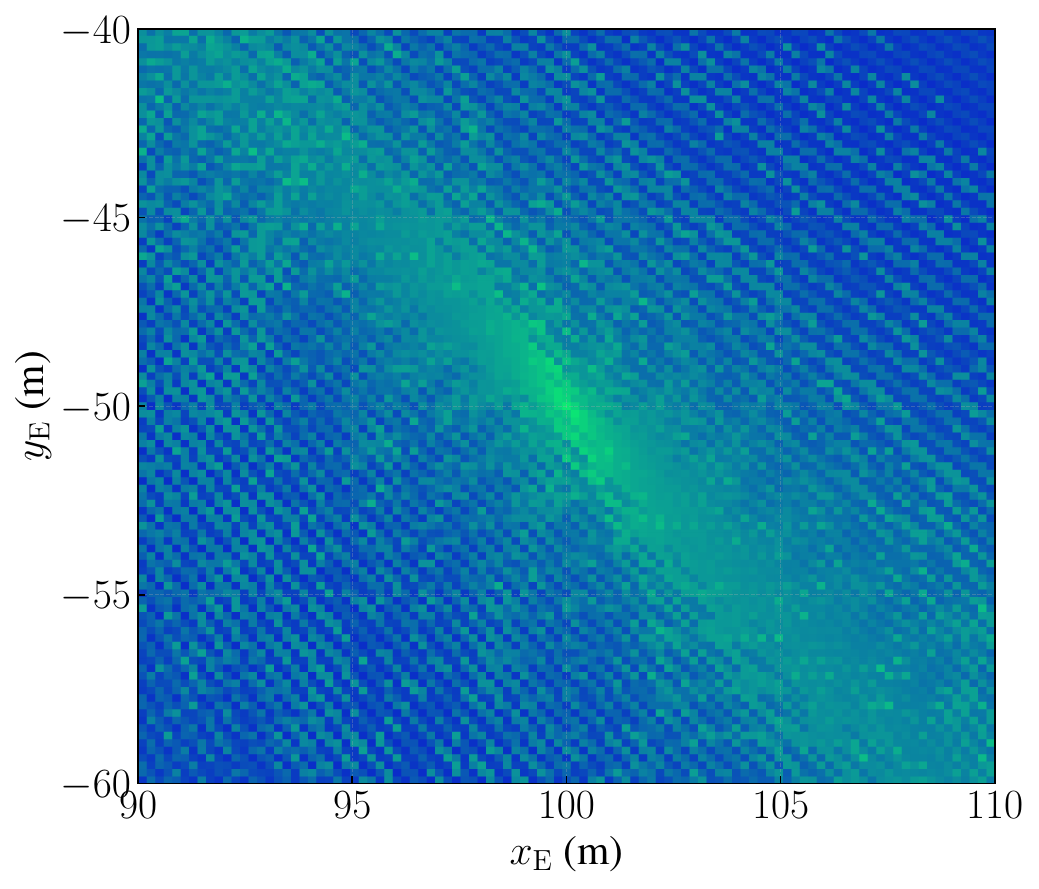}
    \caption{Fixed rotation and optimized position}
    \label{subfig:Different-optimization-variables-c}
\end{subfigure}
\hfill
\begin{subfigure}[t]{0.08\linewidth}
    \centering
    \includegraphics[width=\linewidth]{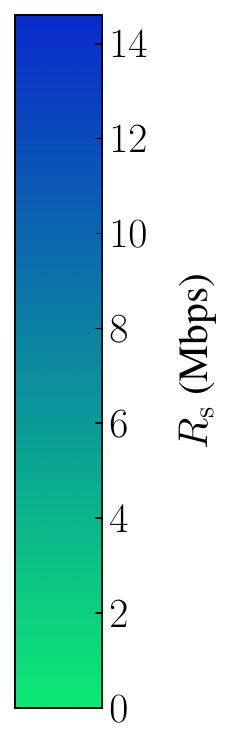}
\end{subfigure}
\caption{Secrecy rate of different deployment schemes for the single-\ac{FR}-aided secure wireless communication system, where the position of Eve, $\boldsymbol{p}_{\mathrm{E}} = [x_\mathrm{E}, y_\mathrm{E},0]^\top$, is varying.}
\label{fig:20-20compare}
\end{figure*}

Fig.~\ref{fig:20-20compare} illustrates the secrecy rate of different deployment schemes for the single-\ac{FR}-aided secure wireless communication system, where the position of Eve is varying.  As shown in Fig.~\ref{subfig:Different-optimization-variables-a}, joint optimization of the position $x_1$ and rotation $\omega_1$ yields the most favorable and spatially robust $R_\mathrm{s}$ distribution. The secrecy rate degrades only when \ac{Eve} lies close to the legitimate receiver \ac{Rx}. In contrast, Fig.~\ref{subfig:Different-optimization-variables-b} fixes the \ac{FR} position at $x_1=50$ and only optimizes the rotation $\omega_1$, resulting in a severely degraded $R_\mathrm{s}$ when \ac{Eve} lies on the line connecting the \ac{FR} and \ac{Rx}. Similarly, Fig.~\ref{subfig:Different-optimization-variables-c} fixes the rotation $\omega=0^{\degree}$ and only optimizes the position $x_1$, preventing the main beam from being steered toward \ac{Rx} in most configurations and thereby yielding a low $R_\mathrm{s}$. These findings confirm that joint position and rotation control is essential to fully leverage the \ac{FR} architecture.

\begin{figure}[!t]
\centering
\includegraphics[width=\linewidth]{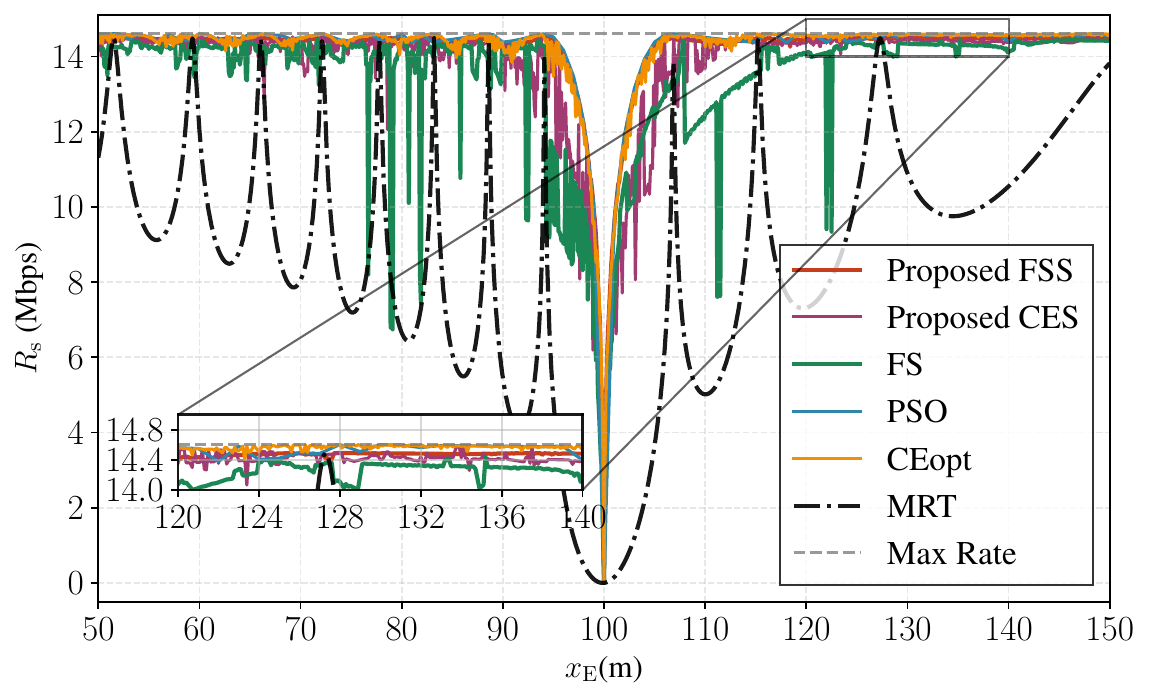}
\caption{Secrecy rate, $R_{\mathrm{s}}$, versus the $x$-coordinates of \ac{Eve}, $x_\mathrm{E}$; performance comparisons of different schemes for the single-\ac{FR}-aided secure wireless communication system.}
\label{fig:x-move-rate}
\end{figure}
Fig.~\ref{fig:x-move-rate} and Table~\ref{tab:time_analysis} present the secrecy rate and runtime comparisons of different schemes for the single-\ac{FR} case, where the $x$ coordinate of \ac{Eve}, $x_\mathrm{E}$, vary from $50$ m to $150$ m. In the legend,  ``FS'' denotes the one-dimensional Fibonacci search method without the subinterval segmentations. ``PSO'' denotes the meta-heuristic particle swarm optimization~\cite{kennedy1995particle}. ``CEopt'' denotes the cross-entropy method for optimization~\cite{rubenstein2004crossentropy}. ``MRT'' denotes the maximum ratio transmission~\cite{lu2025wireless}, which optimizes the positions and angles of the FRs only to maximize the received power at \ac{Rx}, without accounting for \ac{Eve}. ``Max Rate'' indicates the maximum $R_{\mathrm{s}}$ when no eavesdropper is present. It can be observed that our proposed \ac{FSS}, ``PSO'', and ``CEopt'' can achieve near-optimal $R_{\mathrm{s}}$ in single-FR case. The result of ``MRT'' indicates that only enhancing the legitimate received power cannot guarantee communication secrecy, as the achievable secrecy rate depends critically on whether the eavesdropper Eve locates within the beam peak or low-power region. The result of ``FS'' indicates that due to the nonconvex and multimodal nature of the secrecy rate, the direct use of the Fibonacci search method over the entire domain without subintervals segmentation may get trapped in suboptimal solutions and exhibit degraded performance. From Table \ref{tab:time_analysis}, it can be observed that ``PSO'' is the slowest, with an average runtime of $41.82$ s. Moreover, the variance of the runtime is also very large, indicating that ``PSO'' is unstable. By contrast, our proposed CES method drastically reduces the average runtime to merely $0.33$ s and can achieve satisfactory performance in most cases.

\begin{table}[!t]
\renewcommand{\arraystretch}{1.3}
\caption{Runtime Comparisons of Different Schemes for Single-\ac{FR} Case}
\label{tab:time_analysis}
\centering
\begin{tabular}{lccc}
\toprule
\textbf{Method} & \textbf{Mean $\pm$ Std. (s)} & \textbf{Min (s)} & \textbf{Max (s)} \\
\midrule
Proposed FSS          & 7.39 $\pm$ 0.754       & 0.51  & 13.68  \\
Proposed CES         & 0.33 $\pm$ 0.038    & 0.26  & 0.5    \\
FS & 0.031 $\pm$ 0.003  & 0.027 & 0.063  \\
PSO                  & 41.82 $\pm$ 29.56      & 22.08  & 95.99  \\
CEopt                & 2.55 $\pm$ 0.695     & 1.77  & 5.28   \\
\bottomrule
\end{tabular}
\end{table}

\begin{figure}[!htbp]
\centering
\begin{subfigure}{\linewidth}
    \centering
    \includegraphics[width=\linewidth]{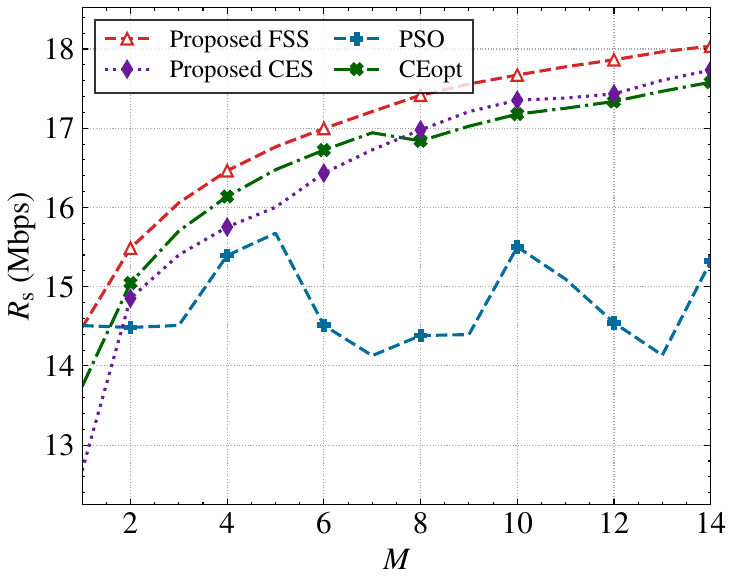}
    \caption{Different optimization methods}
    \label{fig:change_number_method}
\end{subfigure}
\begin{subfigure}{\linewidth}
    \centering
    \includegraphics[width=\linewidth]{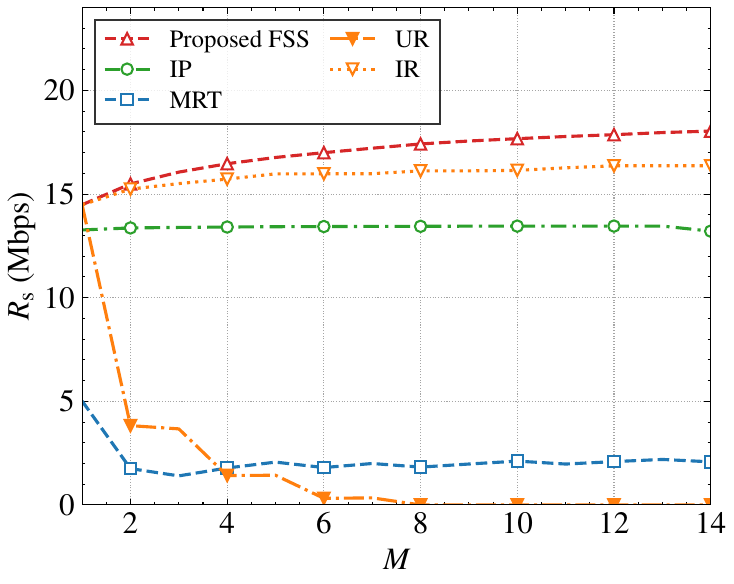}
    \caption{Different deployment schemes}
    \label{fig:change_number_layout}
\end{subfigure}
\caption{Secrecy rate, $R_\mathrm{s}$, versus the number of \ac{FR}, $M$; performance comparisons of different optimization methods and deployment schemes.}
\label{fig:change_number}
\end{figure}

Fig.~\ref{fig:change_number} and Table~\ref{tab:multi_time_analysis} show the secrecy rate and runtime comparisons of different schemes under different numbers of \acp{FR}, $M$. From Fig.~\ref{fig:change_number_method}, it can be observed that our proposed \ac{FSS} method consistently achieves the best performance among all the schemes. The proposed \ac{CES} method achieves comparable performance to ``CEopt''. The performance of ``PSO'' degrades drastically in the multi-FR case, especially when $M$ is large. Table \ref{tab:multi_time_analysis} demonstrates that our proposed \ac{CES} method maintains consistently low computational overhead. In contrast, the runtimes of ``PSO'' and ``CEopt'' grow remarkably with the increase of $M$, indicating the low computational efficiency of these two methods in large-scale optimization. In the legend of Fig.~\ref{fig:change_number_layout}, ``UR'' denotes the scheme that all the FRs use the unified optimized rotation angle with positions fixed. ``IR'' denotes the scheme that all the FRs use the independently optimized rotation angle with positions fixed. ``IP'' denotes the scheme that all the FRs use the independently optimized position with rotations fixed. It can be observed that the maximum achievable secrecy rate is obtained when jointly optimizing the position and rotation. For the ``IR'' scheme, the insufficient \acp{DoF} lead to a lower secrecy rate compared with the fully optimized scheme. For the ``UR'' scheme, the secrecy rate continuously declines as $M$ increases and even approaches zero, since more signal power is leaked to \ac{Eve} due to insufficient control of reflections. For the ``IP'' scheme, since only a single position along the $x$-axis can align the main signal beam towards \ac{Rx}, any additional reflecting surfaces placed at other positions will deliver more signal power to \ac{Eve}, thereby reducing $R_\mathrm{s}$. Moreover, as more \ac{FR}s are deployed, the available space along the $x$-axis becomes limited. Extra \acp{FR} inevitably leak more signal power to \ac{Eve}, resulting in further secrecy rate performance degradation.

\begin{table}[!t]
\renewcommand{\arraystretch}{1.3}
\caption{Runtime Comparisons of Different Schemes for Multi-FR Case}
\label{tab:multi_time_analysis}
\centering
\begin{tabular}{lcccc}
\toprule
\textbf{M} & \textbf{FSS (s)}&\textbf{CES (s)}&\textbf{PSO (s)}&\textbf{CEopt (s)} \\
\midrule
2 & 14.37$\pm$0.14   & 0.38$\pm$0.03  & 178.24$\pm$80.77& 2.90$\pm$0.45  \\
4 & 45.40$\pm$0.37  & 0.73$\pm$0.06  & 507.20$\pm$157.67& 5.51$\pm$0.62 \\
6 & 95.35$\pm$0.49 & 1.07$\pm$0.06  & 674.83$\pm$214.34& 8.00$\pm$0.78 \\
8 & 161.67$\pm$0.67 & 1.42$\pm$0.09  & 900.64$\pm$333.00& 10.55$\pm$0.79 \\
10 & 241.02$\pm$0.86 & 1.79$\pm$0.12  & 1041.66$\pm$219.84& 13.28$\pm$1.11 \\
12 & 327.40$\pm$0.99 & 2.14$\pm$0.16  & 1051.46$\pm$313.61& 15.71$\pm$1.29 \\
14 & 421.81$\pm$1.34 & 2.51$\pm$0.19  & 768.86$\pm$281.54& 17.96$\pm$1.18 \\
\bottomrule
\end{tabular}
\end{table}

\begin{figure}[!htbp]
\centering
\begin{subfigure}{\linewidth}
    \centering
    \includegraphics[width=\linewidth]{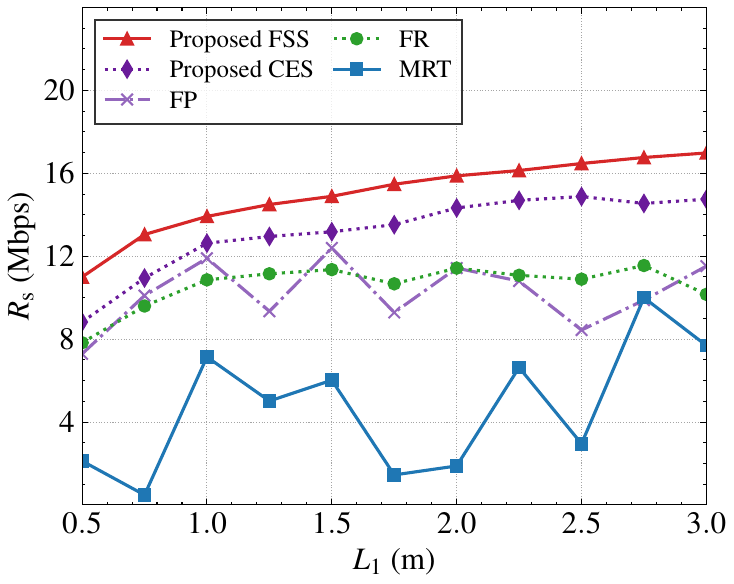}
    \caption{Single-FR}
    \label{fig:change_L1_m1}
\end{subfigure}

\begin{subfigure}{\linewidth}
    \centering
    \includegraphics[width=\linewidth]{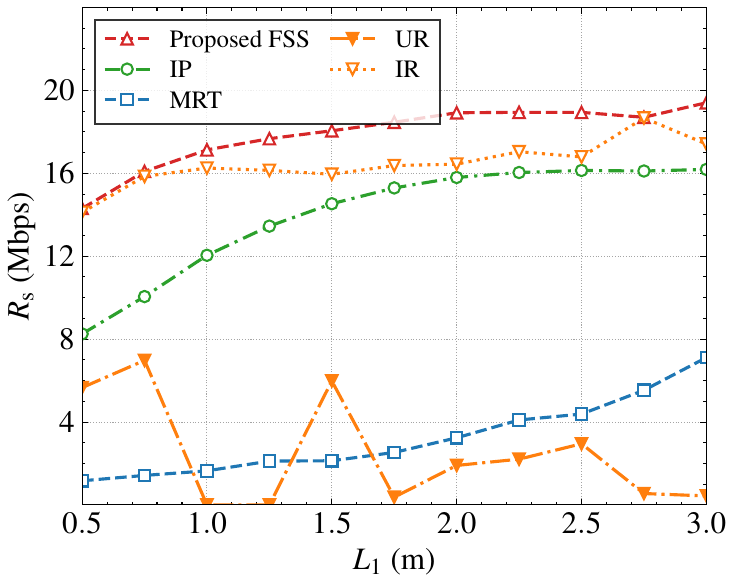}
    \caption{Multi-FR ($M = 10$)}
    \label{fig:change_L1_m10}
\end{subfigure}
\caption{Secrecy rate, $R_\mathrm{s}$, versus the length of FRs, $L_1$; performance comparisons of different deployment schemes for single-\ac{FR} and multi-\ac{FR} cases.}
\label{fig:change_L1}
\end{figure}

Fig.~\ref{fig:change_L1} shows the secrecy rate comparisons of different deployment schemes for single-\ac{FR} and multi-\ac{FR} cases with different lengths of FRs. For the single-\ac{FR} case in Fig.~\ref{fig:change_L1_m1}, in the legend, ``FR'' denotes the scheme using an optimized position with the rotation angle fixed as $\omega_1 = 0\degree$, ``FP'' denotes the scheme using an optimized rotation angle with the position fixed at $x_1 = 50$ m. It can be observed that $R_{\mathrm{s}}$ rises with increasing $L_1$, since a larger \ac{FR} aperture concentrates more signal power around \ac{Rx} and correspondingly reduces the power intercepted by \ac{Eve}. Moreover, it can be also observed that the secrecy rate of the scheme with joint optimization of position and rotation is much higher than the schemes with only position or rotation optimization, which indicates the necessity of joint optimization of both position and rotation. For the multi-FR case in Fig.~\ref{fig:change_L1_m10}, $R_{\mathrm{s}}$ also rises with increasing $L_1$, while the growth rate gradually diminishes until an upper bound is reached. The scheme with joint optimization of position and rotation achieves the highest secrecy rate, while the schemes with only position or rotation optimization yield comparable secrecy rates. This observation indicates that in multi-FR cases, the joint optimization of both position and rotation is not as critical as in single-FR cases, since the additional \acp{DoF} provided by multiple FRs can compensate for the lack of optimization in either position or rotation, especially when the apertures of FRs are large.

\begin{figure}[!htbp]
\centering
\includegraphics[width=\linewidth]{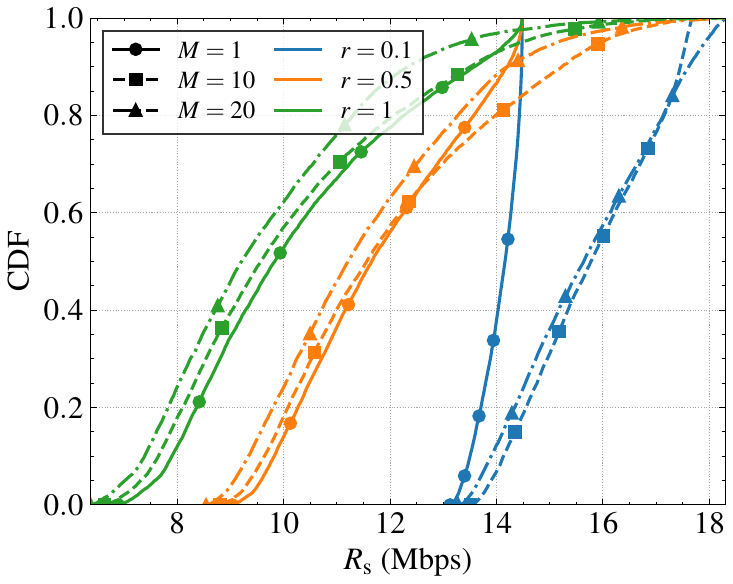}
\caption{CDF of the secrecy rate, $R_\mathrm{s}$, under different numbers of FRs, $M$, and position uncertainty radii of Eve, $r$.}
\label{fig:cdf-compare}
\end{figure}
Fig.~\ref{fig:cdf-compare} illustrates the \ac{CDF} of the secrecy rate under varying numbers of FRs and position uncertainty radii of Eve, where the true position of Eve lies within a circle of radius $r$ centered at $\boldsymbol{p}_{\mathrm{E}}$. It can be observed that $R_\mathrm{s}$ generally improves with an increasing $M$, which is consistent with the results in Fig.~\ref{fig:change_number}. Conversely, $R_\mathrm{s}$ degrades as $r$ increases. This is expected since a larger uncertainty region expands the spatial range of potential Eve positions, thereby increasing the likelihood of Eve being near the main beam. Notably, the impact of adding more FRs depends heavily on the estimation accuracy. When the uncertainty radius is small (e.g., $r=0.1$), increasing $M$ effectively improves the secrecy rate, as the optimized deployment shapes a narrow low-power zone that confines Eve. In contrast, when the uncertainty radius is large (e.g., $r=1$), deploying more \acp{FR} paradoxically degrades performance. This occurs since additional FRs introduce extra beam peaks, and a larger $r$ increases the probability that Eve falls into one of these high-gain regions, leading to a severe signal leakage.

\section{Conclusion}\label{section:conclusion}
This paper investigated the joint position and rotation optimization for \ac{FR}-aided secure wireless communication systems. For the single-\ac{FR} case, we formulated a secrecy rate maximization problem and proposed the \ac{FSS} method, which combines interval partitioning with Fibonacci search. To reduce the computational overhead, we also developed a low-complexity \ac{CES} method. This method uses \ac{GMM}-based adaptive importance sampling to identify promising intervals and then performs Fibonacci search within these intervals. For the multi-\ac{FR} case, after deriving the blocked placement interval incurred by the deployed \acp{FR}, we extended the \ac{FSS} and \ac{CES} methods by deploying the \acp{FR} sequentially. Simulation results showed that the proposed \ac{FSS} method achieves significant secrecy performance gains over all the considered benchmark schemes, whereas the proposed \ac{CES} method provides a more favorable runtime-performance tradeoff. 

The spatial controllability of \acp{FR} also enables two security-related future works. By steering the reflected signal into a deep low-power region around \ac{Eve}, \ac{FR} deployment can support covert communication, where \ac{Eve} cannot detect whether the communication is ongoing. Moreover, the mechanically movable \ac{FR} provides a controllable source of channel randomness. Exploiting channel reciprocity, the Tx and the Rx can obtain shared physical layer secret keys from the \ac{FR}-induced channel variations, whereas \ac{Eve} can only observe uncorrelated measurements.

\appendices   
\section{Proof of Lemma~\ref{lemma2}} 
\label{proof of the lemma 2}
\begin{figure}[!htbp]
\centering
\includegraphics[width=\linewidth]{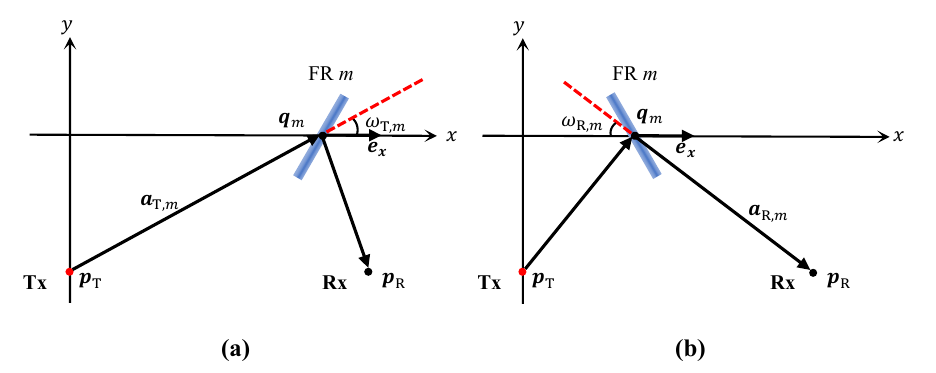}
\caption{Geometric illustration for Tx and Rx lie on different sides of each FR panel.}
\label{fig:lemma2}
\phantomsubcaption\label{fig:lemma2a}
\phantomsubcaption\label{fig:lemma2b}
\end{figure}
Since constraint \eqref{subeq:P3_reflection} enforces that the legitimate receiver \ac{Rx} and the transmitter \ac{Tx} lie on the same side of each \ac{FR}, we analyze two cases separately: $\omega_m > 0$ and $\omega_m < 0$.

For an \ac{FR} with $\omega_m > 0$, we consider the critical rotation angle $\omega_{\mathrm{T},m} \triangleq \arccos({\boldsymbol{a}_{\mathrm{T},m}^\top \boldsymbol{e}_x}/{d_{\mathrm{T},m}})$, defined as the angle between the incident vector $\boldsymbol{a}_{\mathrm{T},m}$ and the $x$-axis, which corresponds to the scenario where the \ac{FR} panel is parallel to $\boldsymbol{a}_{\mathrm{T},m}$. As illustrated in Fig.~\ref{fig:lemma2a}, when $\omega_m > \omega_{\mathrm{T},m} $, \ac{Tx} and \ac{Rx} will no longer reside on the same side of the \ac{FR}.

Similarly, for an \ac{FR} with $\omega_m < 0$, we consider the critical rotation angle $\omega_{\mathrm{R},m} \triangleq \arccos({\boldsymbol{a}_{\mathrm{R},m}^\top \boldsymbol{e}_x}/{d_{\mathrm{R},m}})$ corresponding to the scenario where the \ac{FR} panel is parallel to the reflected vector $\boldsymbol{a}_{\mathrm{R},m}$. As illustrated in Fig.~\ref{fig:lemma2b}, when $\omega_m < -\omega_{\mathrm{R},m} $, \ac{Tx} and \ac{Rx} will no longer reside on the same side of the \ac{FR}.

Therefore, when constraint \eqref{subeq:P3_reflection} holds, the feasible range of the corresponding rotation angle $\omega$ is given by
\begin{equation}\label{eq:omega_range_T,R}
-\omega_{\mathrm{R},m} \leq \omega \leq \omega_{\mathrm{T},m}.
\end{equation}

\section{Proof of Lemma~\ref{lemma3}} 
\label{proof of the lemma 3}

\begin{figure}[!htbp]
\centering
\includegraphics[width=\linewidth]{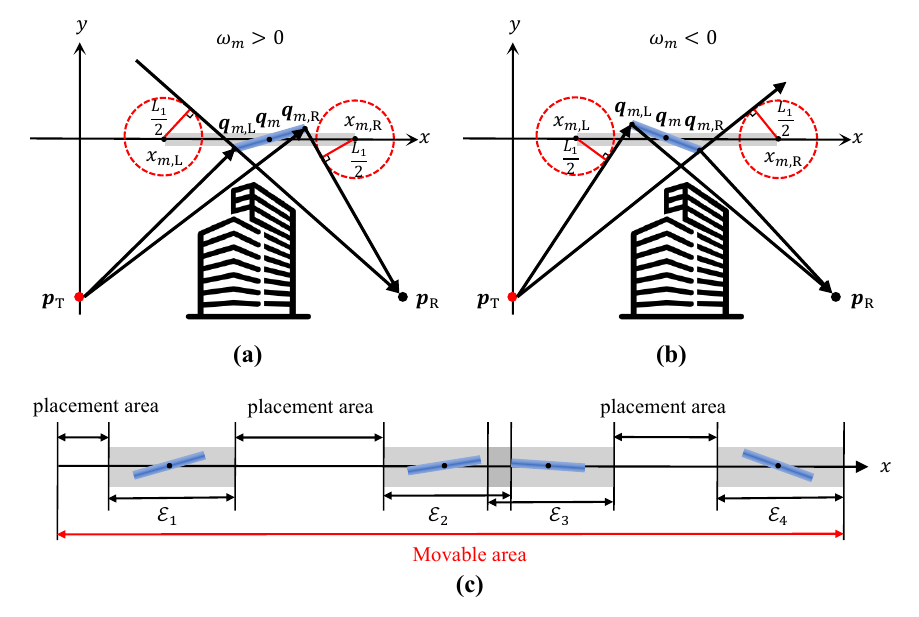}
\caption{Blocked placement interval induced by a deployed FR in the multi-FR case.}
\label{fig:lemma3}
\phantomsubcaption\label{fig:lemma3a}
\phantomsubcaption\label{fig:lemma3b}
\end{figure}

For a line $\ell:\ \boldsymbol{n}^\top \boldsymbol{x} + c_0 = 0$, which passes through two points $\boldsymbol{p} = [p_x,\, p_y,0]^\top$ and $\boldsymbol{q} = [q_x,\, q_y,0]^\top$ with $\boldsymbol{x} = [x,\, y, \, 0]^\top$, the normal vector and the line constant are given by
\begin{equation*}
\boldsymbol{n} = \begin{bmatrix} n_x \\ n_y \\ 0 \end{bmatrix} = \begin{bmatrix} q_y - p_y \\ p_x - q_x \\ 0 \end{bmatrix}, \qquad c_0 = q_x p_y - p_x q_y.
\end{equation*}
The distance from a point on the $x$-axis, $[x,0,0]^\top$, to $\ell$ is
\begin{equation}\label{eq:signed_distance}
\operatorname{dist}\bigl(([x,0,0]^\top),\,\ell\bigr) = \frac{|n_x x + c_0|}{\|\boldsymbol{n}\|}.
\end{equation}
Setting $\operatorname{dist} = L_1/2$ and solving for $x$ yields the two intersections of the $x$-axis with the pair of offset lines at distance $L_1/2$ from $\ell$, i.e.,
\begin{equation}\label{eq:general_boundary}
x = \frac{\pm \dfrac{L_1}{2}\|\boldsymbol{n}\| - c_0}{n_x}.
\end{equation}
The two roots are the left and right boundary candidates on the $x$-axis.

As illustrated in Fig.~\ref{fig:lemma3a}, for the case $\omega_m > 0$, the left-boundary line passes through the left \ac{FR} endpoint $\boldsymbol{q}_{m,\mathrm{L}}$ and the receiver $\boldsymbol{p}_{\mathrm{R}}$, and the right-boundary line passes through the right \ac{FR} endpoint $\boldsymbol{q}_{m,\mathrm{R}}$ and the receiver $\boldsymbol{p}_{\mathrm{R}}$. Applying \eqref{eq:general_boundary} and combining the minimum distance requirement $d_{\min}$, we obtain 
\begin{align*}
x_{m,\mathrm{L}} &= \min\!\left\{ \frac{\pm\frac{L_1}{2}\| \boldsymbol{n}_{\mathrm{L}}\| - c_{0,\mathrm{L}}}{n_{\mathrm{L},x}}, x_m-d_{\min} \right\},\\
x_{m,\mathrm{R}} &= \max\!\left\{ \frac{\pm\frac{L_1}{2}\| \boldsymbol{n}_{\mathrm{R}}\| - c_{0,\mathrm{R}}}{n_{\mathrm{R},x}}, x_m+d_{\min} \right\},
\end{align*}
where $\boldsymbol{n}_{\mathrm{L}} = \begin{bmatrix} y_{m,\mathrm{L}} - y_{\mathrm{R}} \\ x_{\mathrm{R}} - x_{m,\mathrm{L}} \end{bmatrix}$, $\boldsymbol{n}_{\mathrm{R}} = \begin{bmatrix} y_{m,\mathrm{R}} - y_{\mathrm{R}} \\ x_{\mathrm{R}} - x_{m,\mathrm{R}} \end{bmatrix}$, $c_{0,\mathrm{L}} =  x_{m,\mathrm{L}} y_{\mathrm{R}}-x_{\mathrm{R}} y_{m,\mathrm{L}} $, $c_{0,\mathrm{R}} =  x_{m,\mathrm{R}} y_{\mathrm{R}}-x_{\mathrm{R}} y_{m,\mathrm{R}} $, $n_{\mathrm{L},x} = y_{m,\mathrm{L}} - y_{\mathrm{R}}$, and $n_{\mathrm{R},x} = y_{m,\mathrm{R}} - y_{\mathrm{R}}$.

Similarly, as illustrated in Fig.~\ref{fig:lemma3b}, for the case $\omega_m < 0$, the left-boundary line passes through the left \ac{FR} endpoint $\boldsymbol{q}_{m,\mathrm{L}}$ and the transmitter $\boldsymbol{p}_{\mathrm{T}}$, and the right-boundary line passes through the right \ac{FR} endpoint $\boldsymbol{q}_{m,\mathrm{R}}$ and the transmitter $\boldsymbol{p}_{\mathrm{T}}$. Applying \eqref{eq:general_boundary} and combining the minimum distance requirement $d_{\min}$, we obtain
\begin{align*}
x_{m,\mathrm{L}} &= \min\!\left\{ \frac{\pm\frac{L_1}{2}\| \boldsymbol{n}_{\mathrm{L}}\| - c_{0,\mathrm{L}}}{n_{\mathrm{L},x}}, x_m-d_{\min} \right\},\\
x_{m,\mathrm{R}} &= \max\!\left\{ \frac{\pm\frac{L_1}{2}\| \boldsymbol{n}_{\mathrm{R}}\| - c_{0,\mathrm{R}}}{n_{\mathrm{R},x}}, x_m+d_{\min} \right\},
\end{align*}
where $\boldsymbol{n}_{\mathrm{L}} = \begin{bmatrix} y_{m,\mathrm{L}} - y_{\mathrm{T}} \\ x_{\mathrm{T}} - x_{m,\mathrm{L}} \end{bmatrix}$, $\boldsymbol{n}_{\mathrm{R}} = \begin{bmatrix} y_{m,\mathrm{R}} - y_{\mathrm{T}} \\ x_{\mathrm{T}} - x_{m,\mathrm{R}} \end{bmatrix}$, $c_{0,\mathrm{L}} =  x_{m,\mathrm{L}} y_{\mathrm{T}} - x_{\mathrm{T}} y_{m,\mathrm{L}}$, $c_{0,\mathrm{R}} =  x_{m,\mathrm{R}} y_{\mathrm{T}} - x_{\mathrm{T}} y_{m,\mathrm{R}}$, $n_{\mathrm{L},x} = y_{m,\mathrm{L}} - y_{\mathrm{T}}$, and $n_{\mathrm{R},x} = y_{m,\mathrm{R}} - y_{\mathrm{T}}$.

\bibliographystyle{IEEEtran}
\bibliography{Secure_FR_2026TIFS,IEEEtranBSTCTL}
\end{document}